\documentclass[lettersize,journal]{IEEEtran}
\usepackage{cite}
\usepackage[utf8]{inputenc} 
\usepackage[T1]{fontenc}    
\usepackage{hyperref}       
\usepackage{url}            
\usepackage{booktabs}       
\usepackage{amsfonts}       
\usepackage{nicefrac}       
\usepackage{microtype}      
\usepackage{doi}
\usepackage{amsmath,amssymb,amsfonts}
\usepackage{algorithmic}
\usepackage{graphicx}
\usepackage{textcomp}
\usepackage{xcolor}
\usepackage{bbm}
\usepackage{amsthm}
\usepackage{setspace}
\usepackage{physics}
\usepackage{algorithm}
\usepackage{mathtools}
\usepackage{diagbox}
\usepackage{caption}
\usepackage{subcaption}
\usepackage{physics}
\usepackage{multirow}
\usepackage{epstopdf}
\usepackage{tabularx}
\usepackage{ulem}

\newcommand{\etal}{\textit{et al.}}

\makeatletter
\newcommand{\distas}[1]{\mathbin{\overset{#1}{\kern\z@\sim}}}%
\newsavebox{\mybox}\newsavebox{\mysim}
\newcommand{\distras}[1]{%
  \savebox{\mybox}{\hbox{\kern3pt$\scriptstyle#1$\kern3pt}}%
  \savebox{\mysim}{\hbox{$\sim$}}%
  \mathbin{\overset{#1}{\kern\z@\resizebox{\wd\mybox}{\ht\mysim}{$\sim$}}}%
}
\makeatother
\makeatletter
\newcommand*\rel@kern[1]{\kern#1\dimexpr\macc@kerna}
\newcommand*\widebar[1]{%
  \begingroup
  \def\mathaccent##1##2{%
    \rel@kern{0.8}%
    \overline{\rel@kern{-0.8}\macc@nucleus\rel@kern{0.2}}%
    \rel@kern{-0.2}%
  }%
  \macc@depth\@ne
  \let\math@bgroup\@empty \let\math@egroup\macc@set@skewchar
  \mathsurround\z@ \frozen@everymath{\mathgroup\macc@group\relax}%
  \macc@set@skewchar\relax
  \let\mathaccentV\macc@nested@a
  \macc@nested@a\relax111{#1}%
  \endgroup
}
\makeatother
\newtheorem{definition}{Definition}
\newtheorem{remark}{Remark}
\newtheorem{assumption}{Assumption}
\newtheorem{theorem}{Theorem}
\newtheorem{lemma}{Lemma}

\renewcommand{\algorithmiccomment}[1]{// \quad #1}

\begin{document}
\title{Online Influence Maximization with Semi-Bandit Feedback under Corruptions}
\author{Xiaotong Cheng, Behzad Nourani-Koliji, Setareh Maghsudi\thanks{
X. Cheng is with the Department of Electrical Engineering and Information Technology, Ruhr-University Bochum, 44801 Bochum, Germany (email:xiaotong.cheng@ruhr-uni-bochum.de). B. Nourani-Koliji is with the Department of Computer Science, University of Tübingen, 72074 Tübingen, Germany (email:behzad.nourani-koliji@uni-tuebingen.de). S. Maghsudi is with the Department of Electrical Engineering and Information Technology, Ruhr-University Bochum, 44801 Bochum, Germany and with the Fraunhofer Heinrich Hertz Institute, 10587 Berlin, Germany (email:setareh.maghsudi@rub.de).}\\
}
\markboth{Journal of \LaTeX\ Class Files,~Vol.~14, No.~8, August~2021}%
{Shell \MakeLowercase{\textit{et al.}}: A Sample Article Using IEEEtran.cls for IEEE Journals}
\maketitle
\begin{abstract}
In this work, we investigate the online influence maximization in social networks. Most prior research studies on online influence maximization assume that the nodes are fully cooperative and act according to their stochastically generated influence probabilities on others. In contrast, we study the online influence maximization problem in the presence of some corrupted nodes whose damaging effects diffuse throughout the network. We propose a novel bandit algorithm, CW-IMLinUCB, which robustly learns and finds the optimal seed set in the presence of corrupted users. Theoretical analyses establish that the regret performance of our proposed algorithm is better than the state-of-the-art online influence maximization algorithms. Extensive empirical evaluations on synthetic and real-world datasets also show the superior performance of our proposed algorithm.  
\end{abstract}
%
\section{Introduction}
\label{sec:intro}
In the last decade, social networks have played critical roles in the analysis and optimization of data in epidemiology, marketing, and economics \cite{easley2010networks, cesa2013gang, valko2016bandits}, as a result of information propagation or diffusion that is inherent in such networks. That has led to developing frameworks such as Influence Maximization (IM), where companies aim to select a fixed number of customers that greatly influence others, called seeds or source nodes, to receive reimbursement in return for advertising their products \cite{wen2017online,vaswani2017model,Stein2017HeuristicAF}. The companies aim at maximizing the influence spread given a limited budget. 

Online social networks enable frequent information collection and updating of information regarding users' connections and interactions, which simplifies IM. In most solutions, a social network is modeled as a graph, and users as nodes. The edges represent the users' relations and the edge weights represent influence probabilities between users. Influence propagates through the network under a specific diffusion model. The independent cascade (IC) model and linear threshold (LT) model are the two most widely used models \cite{wu2019factorization,li2020online}. In the IC model, an adopter user has an activation probability, i.e., to convince each neighbor to adopt the product. The activation probabilities between pairs of users are independent. In the LT model, a user adopts the product only if the aggregated influence from its neighbors reaches a threshold \cite{kempe2003maximizing}. The IC model is particularly well-known and frequently studied, especially in the context of online influence maximization  \cite{wen2017online,wu2019factorization}. Therefore, in this paper, we focus on the IM problem within the framework of the IC model.

In the offline IM problem, the network structure and edge weights are known in advance \cite{kempe2003maximizing}. However, in real applications, even if the network topology is accessible, the influence probabilities are unknown a priori. That highlights the importance of online influence maximization (OIM) problem \cite{vaswani2017model, wen2017online, wu2019factorization, lei2015online}. In OIM, the activation probability is unknown and needs to be estimated by a learner through directly interacting with the network.

The researchers have studied the OIM problem from many perspectives \cite{wen2017online,vaswani2017model,wu2019factorization,li2020online,zuo2022online}; Nevertheless, they mostly assume that all users in the social networks are fully cooperative and influence others voluntarily and automatically, which ignores the adverse effects of potentially corrupted users/nodes as a critical factor. However, in real-world applications, malicious users trick the system with disputed behaviors. Even if not selected as seeds, they can spread corruption effects throughout the system by disrupting the information flow. Hence, it is imperative to develop an algorithm to address this challenge in influence maximization. 

Before describing our contributions, we motivate our settings with several examples. Nowadays, customers heavily lean on online reviews to guide their purchasing decisions. These reviews extend beyond traditional online shopping platforms like Amazon. They also play a significant role in invisible marketing, where brands collaborate with influencers to integrate the products into their content. However, not all reviews are persuasive; sometimes, they can have a subtle counterproductive effect \cite{liu2002interactivity,fitzsimons2004reactance,dore2018controversial}, e.g., when one uses humor or puns to subtly highlight a product's potential weaknesses, thus reducing the followers' enthusiasm to purchase that product. Another example is when the influencers adopt the comparative method for recommendation, which draws the customers' attention to similar products. Overemphasis of the products is another instance \cite{burgoon2002revisiting}. Such behaviors, malicious or not, impact the activation probabilities. In all these cases, most activation probabilities still follow a predictable pattern, whereas a fraction of them are corrupted under arbitrary patterns and are not identically distributed over time.

The proposed corruption-robust IM is not a straightforward extension of the previous work on corruption-robust bandit algorithms. While the concept of corruption-robust bandit algorithms is not new in the research of linear bandits, its application within the online influence maximization remains unexplored. Although each user in OIM can represent an arm in the bandit setting, it cannot be directly generalized to combinatorial setting since the seeds are not selected in isolation as the users mutually affect each other according to the social network model. Besides, in OIM, the reward is not a linear function of the outcomes obtained from each selected seed and has a more complicated structure. It involves the cascading feedback model and limited feedback information (binary feedback). Furthermore, IM introduces a unique aspect where the impact of corruption can also propagate throughout the entire network, which makes the problem even more challenging. Additionally, the offline IM with given graph and activation probability information is an NP-hard problem. 

In this work, we develop a novel OIM framework within a social network with several corrupted users to fill the gap of OIM problem under corruption. We summarize our main contributions as follows.
\begin{itemize}
    \item We propose a novel algorithm for corruption-robust OIM, titled CW-IMLinUCB, which builds on an OIM algorithm with a corruption-robust linear bandit algorithm \cite{he2022nearly}. By integrating the weighted regression into an OIM algorithm, our proposal alleviates the problems arising from inaccurate estimations caused by corrupted users.
    \item We theoretically demonstrate that our proposed CW-IMLinUCB algorithm achieves the regret guarantee $O(dBE^*\sqrt{T}\log(nT) + BE^*E^cCd\log(nT))$ while being robust to malicious behaviors. 
    \item Extensive experiments on synthetic and real-world datasets show the superior performance of our algorithm compared with the existing methods.
\end{itemize}
The rest of this paper is organized as follows: Section~\ref{sec:related-work} summarizes the related work. In Section~\ref{sec:problem}, we formulate the online influence maximization problem under corruption. Section~\ref{sec:cwb} introduces our proposed algorithm CW-IMLinUCB. Section~\ref{sec:ta} presents the theoretical analysis of CW-IMLinUCB and Section~\ref{sec:experiments} demonstrates the experimental results. Finally, Section~\ref{sec:conclusion} concludes this work.

\section{Related Work}\label{sec:related-work}
Our work is related to IM and bandits with adversarial corruptions.

\subsection{IM Related Work} 
Influence maximization is first investigated as an algorithmic problem in \cite{richardson2002mining}.
Reference \cite{kempe2003maximizing} formulates the influence maximization problem as a discrete optimization problem and proves the problem to be NP-hard. A greedy approximation algorithm is proposed and shown to be effective for both IC and LT models. The efficiency of this greedy algorithm is further improved in \cite{chen2009efficient}. Besides, reference \cite{gao2020fair} extends the IM problem to competitive influence maximization across multiple social events. In \cite{khatri2023cks}, the proposed approach solves the IM problem by identifying community bridge nodes and select them as seed set. Aforementioned work considers the offline IM problem setting, where the network structure and edge weights are known in advance \cite{kempe2003maximizing}. However, in realistic scenarios, even if the topology of the social network might be known, via Facebook, or Twitter, etc, the influence probabilities are unknown apriori. This highlights the importance of online influence maximization (OIM) problem \cite{vaswani2015influence, vaswani2017model, wen2017online, wu2019factorization, lei2015online}. 

The framework of OIM problem can be formulated as a variation of combinatorial multi-armed bandit (CMAB) problem, where the learning agent selects several base arms, defined as super arm at each round and tries to maximize its cumulative reward \cite{chen2013combinatorial,chen2016combinatorial}. References \cite{chen2013combinatorial,chen2016combinatorial} first use the CMAB framework to solve the OIM problem and develop an `Upper Confidence Bound (UCB)'-like algorithm based on the IC model and analyze the regret bound. In \cite{wen2017online}, the authors consider the OIM problem with an independent cascade semi-bandit (ICSB) model. They propose the linear generalization model of the activation probability and prove regret bounds assuming edge-semi bandit feedback. Reference \cite{vaswani2017model} suggests a different parameterization for the IM problem concerning pairwise reachability probabilities regardless of the underlying diffusion models. In \cite{wu2019factorization}, the authors factorize the activation probability on the edges into two latent factors. They use an IC model to estimate the influence parameters at the node level. There are a few works investigating OIM under different diffusion models. Reference \cite{li2020online} presents the OIM problem under the LT model. Wu \etal \cite{wu2020evolving} address the non-stationarity in an evolving underlying social network whose nodes and edges change over time. Reference \cite{zuo2022online} introduce competitive concept into OIM problem and extend the classical IC model to multi-item diffusion model. Authors in \cite{10.5555/3545946.3598895} study the OIM problem under a decreasing cascade model, which is also a variation of IC model with consideration of market saturation. Similar to \cite{wen2017online,wu2019factorization,zuo2022online,10.5555/3545946.3598895}, we use an IC model for influence propagation under edge-level feedback, while the corruption in diffusion is also considered. 

\subsection{Bandits with Corruption Related Work} Reference \cite{lykouris2018stochastic} extends the classic stochastic multi-armed bandit problem by allowing for corrupted feedback and developing a decision-making strategy whose regret is proportional to the total corruption at each round. In \cite{gupta2019better}, the authors propose an algorithm for a similar setting, whose regret is the summation of two terms: a corruption-independent term that matches the regret of the seminal multi-armed bandit algorithm and a time-independent term that is linear in the total corruption. For the corrupted stochastic linear bandit setting, Li \etal \cite{li2019stochastic} present an algorithm with an instance-dependent regret bound. For the same problem, the algorithm in \cite{bogunovic2021stochastic} achieves a regret with a corruption term that is linear in the total corruption. Zhao \etal \cite{zhao2021linear} develop a variance-aware algorithm based on the `Optimism in the Face of Uncertainty Linear bandit (OFUL)' algorithm \cite{abbasi2011improved}. Reference \cite{he2022nearly} also proposes a computationally efficient algorithm based on OFUL, by incorporating a weighted ridge regression that prevents using the contexts whose rewards might be corrupted. Wang \etal \cite{wang2023online} extend the work of \cite{he2022nearly} by considering the online clustering bandits problem with corrupted users. Besides, an algorithm is proposed to identify such users. In our work, we focus on the weighted ridge regression approach utilized in \cite{he2022nearly,wang2023online}. 

\section{Problem Setup}
\label{sec:problem}
In this section, we formulate an OIM problem under corruption with bandit feedback, illustrated in Figure~\ref{fig:oim-cor}. In such a problem, the final corruption effect depends on the position of the corrupted users. Indeed, a higher probability of activating corrupted users increases the corruption level in the system. Sometimes, even a tiny perturbation by a corrupted node in some time intervals changes the seed set entirely, thereby failing a corruption-agnostic agent. That aspect differs fundamentally from previous work on corruption \cite{he2022nearly,wang2023online}. 
\begin{figure}[!htp]
  \centering
  \includegraphics[width=0.99\linewidth]{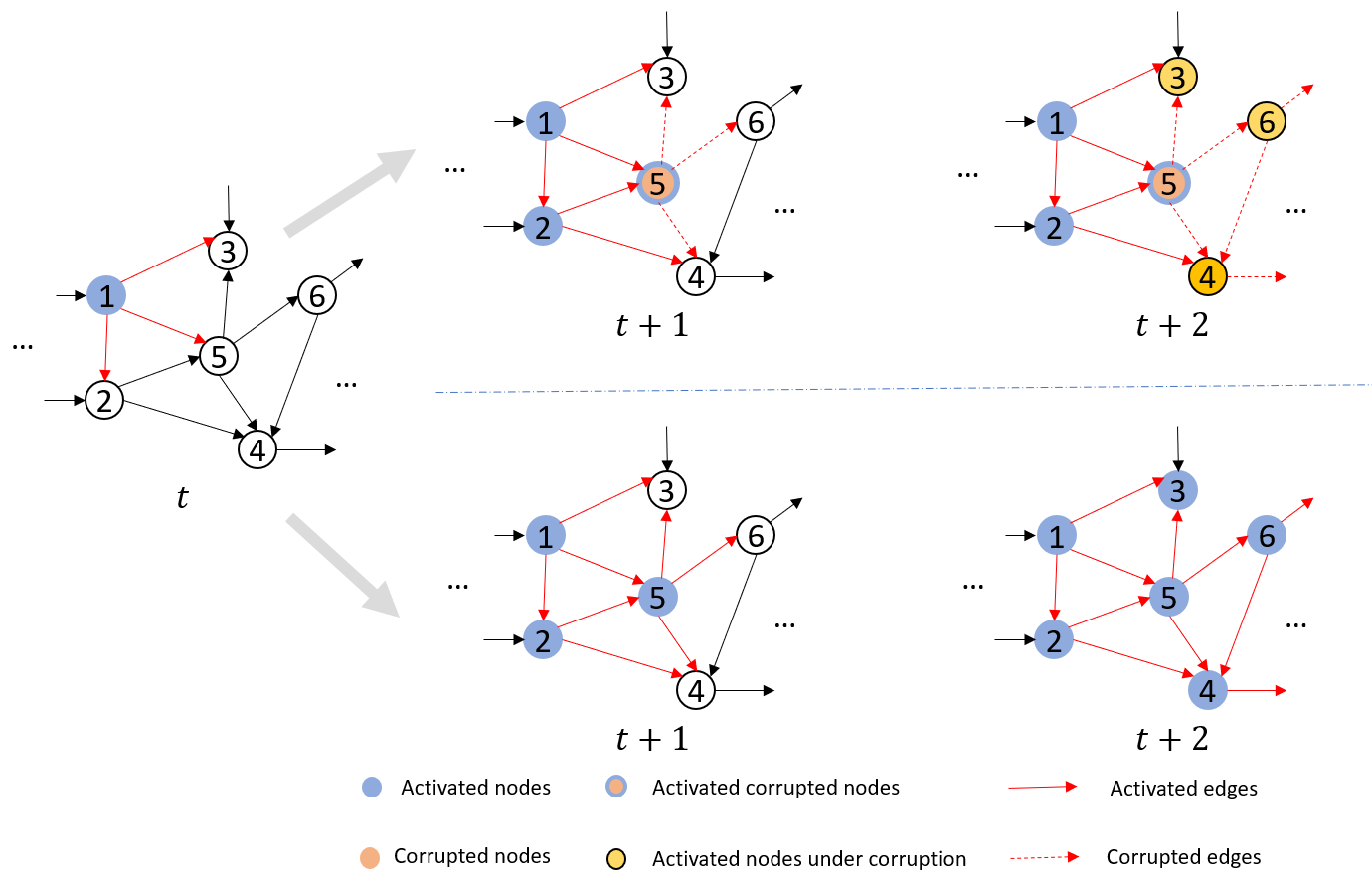}
  \caption{Online influence maximization under corruption. \\ Edges between users represent potential pathways for influence propagation. At $t$, when User 1 is activated, all of its out-edges trigger the activation of connected users. However, if User 5 is a corrupted user with unpredictable behavior, the outcomes at $t+1$ can vary. With some unknown probability, User 5 may behave normally (depicted in blue) or act adversarially (depicted in orange). User 5's corrupted behavior does not only perturb the influence diffusion when one selects it as a seed but also interferes with the influence diffusion when Users 1 or 2 are seeds, as user 5 lies within their diffusion pathways.}
  \label{fig:oim-cor}
\end{figure} 
\subsection{Notation}
We use boldface lowercase letters and boldface uppercase letters to represent vectors and matrices respectively. For example, $\norm{\boldsymbol{x}}_p$ denotes the $p-$norm of a vector $\boldsymbol{x}$. For a symmetric positive semi-definite matrix $\boldsymbol{A} \in \mathbb{R}^{d \times d}$, the weighted $2-$norm of vector $\boldsymbol{x} \in \mathbb{R}^d$ is defined by $\norm{\boldsymbol{x}}_{\boldsymbol{A}} = \sqrt{\boldsymbol{x}^T\boldsymbol{A}\boldsymbol{x}}$. The inner product is denoted by $\langle \cdot, \cdot \rangle$ and the weighted inner-product $\langle \boldsymbol{x},\boldsymbol{y} \rangle_{\boldsymbol{A}} = \boldsymbol{x}^T \boldsymbol{A} \boldsymbol{y}$. \textbf{Table~\ref{tab:nota}} summarizes the definitions and notations.
\begin{table}[!ht]
\begin{center}
\centering
\captionsetup{justification=centering}
\caption{Notation}
\label{tab:nota}
 \begin{tabular}{c p{7cm}}
 \hline
 \multicolumn{2}{l}{\textbf{Problem-specific notations}} \\
 \hline
 $n$ & Number of users \\
 $m$ & Number of edges \\
 $\mathcal{G}$ & Graph that models the social network \\
 $\mathcal{V}$ & User set \\
 $\mathcal{E}$ & Edge set \\
 $\mathcal{S}_t$ & Seed set selected at $t$\\
 $K$ & Budget of seed set \\
 $p(e)$ & Activation probability of edge $e$ \\
 $d$ & Dimension of feature vectors \\
 $T$ & Total number of rounds \\
 $\boldsymbol{x}_e$ & Feature vector of edge $e$ \\
 $\boldsymbol{\theta}$ & Unknown feature vector \\
 $c_{u,t}$ & corruption level of node $u$ at $t$ \\
 $C_u$ & Total corruption budget of node $u$ \\
 $C$ & The maximum corruption budget \\
 $\omega_{e,t}$ & Weight coefficient of edge $e$ \\
 $\boldsymbol{M}_t$ & Gram matrix \\
 $\boldsymbol{b}_t$ & Vector that summarizes the past propagations \\
 $\hat{\boldsymbol{\theta}}$ & Estimation of unknown feature vector $\boldsymbol{\theta}$ \\
 \hline
 \end{tabular}
 \end{center}
\end{table}

\subsection{Influence Maximization}
In the influence maximization (IM) problem, a directed graph $\mathcal{G} = (\mathcal{V}, \mathcal{E})$ is utilized to model the social network. $\mathcal{V} = \{1,2,\ldots,n\}$ is the set of users (nodes) and $\mathcal{E}$ is the set of edges with cardinality $m = |\mathcal{E}|$. Each edge $e \in \mathcal{E}$ is associated with an activation probability $p(e) \in [0,1]$. For example, an edge $e = (u,v)$ represents that user $v$ follows user $u$ on some social media and $p(u,v)$ represents the probability that user $v$ (receiving node) will be activated/influenced by user $u$ (giving node). Denote $\boldsymbol{P} = (p(e_1),\ldots,p(e_m))$ to be the activation probability vector. For a given seed set $\mathcal{S} \subseteq \mathcal{V}$ with activation probability $\boldsymbol{P}$, the expected number of influenced users under the diffusion model $D$ is $f_{D,\boldsymbol{P}}(\mathcal{S})$. By definition, the users/nodes in $\mathcal{S}$ are always influenced. 

Given $\mathcal{G}$ and a budget $K$ on the number of seeds to be selected, IM aims to find a seed set that maximizes the influence spread. Formally,
\begin{gather}
      \mathcal{S}^{\text{opt}} = \arg \max_{|\mathcal{S}|\leq K}f_{D,\boldsymbol{P}}(\mathcal{S}).
\end{gather}
The IM problem is NP-hard \cite{wen2017online,vaswani2017model,wu2019factorization}, but approximation algorithms exist \cite{chen2010scalable,tang2014influence}. In this paper, we refer to such algorithms as \textit{oracles}, which take a graph, size of the seed set, and activation probabilities of all edges as inputs and output an appropriate set of seeds.

Define $\mathcal{S}^{\text{opt}}$ as the optimal solution of the problem and $\mathcal{S}^* = \textsc{ORACLE}(\mathcal{G},K,\boldsymbol{P})$ as the (possibly random) solution of an oracle \textsc{ORACLE}. It serves as an $(\alpha,\gamma)$-approximation of $\mathcal{S}^{\text{opt}}$, $\alpha, \gamma \in [0,1]$, where $f_{D,\boldsymbol{P}}(\mathcal{S}^*) \geq \gamma f_{D,\boldsymbol{P}}(\mathcal{S}^{\text{opt}})$ with probability at least $\alpha$ \cite{chen2016combinatorial}. This further implies that $\mathbb{E}[f_{D,\boldsymbol{P}}(\mathcal{S}^*)] \geq \alpha \gamma f_{D,\boldsymbol{P}}(\mathcal{S}^{\text{opt}})$. Besides, if $\alpha = \gamma = 1$, the oracle is exact. 

The OIM problem is approachable within CMAB framework. In such a model, the users and any seed set represent the arms and a super arm, respectively. We assume that the expected influence spread (expected reward) satisfies the following assumptions, which are standard assumptions also in combinatorial bandit problems \cite{wang2017improving,wu2019factorization}. 
\begin{assumption}[Monotonicity]
\label{assump:mono}
The expected reward of playing any super arm $\mathcal{S} \in \bar{\mathcal{S}}$ is monotonically non-decreasing with respect to the expectation vector, i.e., if for all $i \in [m]$, $p(e_i) \leq p'(e_i)$, we have $f_{D,\boldsymbol{P}}(\mathcal{S}) \leq f_{D,\boldsymbol{P}'}(\mathcal{S})$, for all $\mathcal{S} \in \bar{\mathcal{S}}$ with $\bar{\mathcal{S}}$ being the set of all candidate super arms.
\end{assumption}
\begin{assumption}[1-Norm Bounded Smoothness]
\label{assump:bound}
A combinatorial multi-armed bandit with probabilisitically triggered arms (CMAB-T) satisfy 1-norm bounded smoothness, if there exists a bounded smoothness constant $B \in \mathbb{R}^{+}$ such that for any two distributions with expectation vectors $\boldsymbol{P}$ and $\boldsymbol{P}'$ and any action $\mathcal{S}$, we have $|f_{D,\boldsymbol{P}} - f_{D,\boldsymbol{P}'}| \leq B \sum_{i \in \tilde{S}}|p(e_i) - p'(e_i)|$, where $\tilde{S}$ is the set of edges (arms) that are triggered by $\mathcal{S}$. 
\end{assumption}
\begin{remark}
For the OIM problems, the 1-Norm Bounded Smoothness (Assumption~\ref{assump:bound}) holds with smoothness constant $B = \tilde{n}$, where $\tilde{n}$ is the largest
number of nodes any node can reach in the directed graph $\mathcal{G} = (\mathcal{V},\mathcal{E})$ \cite{wang2017improving}.
\end{remark}

\subsection{Online Influence Maximization under Corruption}
In real-world applications, the activation probability vector $\boldsymbol{P}$ is unknown and shall be learned via interaction with the network: At each round $t$, the learner/agent firstly chooses a seed set $\mathcal{S} \subseteq \mathcal{V}$ with cardinality $K$ based on its prior information and past observations. It then uses the feedback from the observed influence spread to refine the estimation of $\boldsymbol{P}$. The learner aims to maximize the influence spread through this repeated process. Multi-armed bandit framework, especially the linear bandit model, is widely used to solve the OIM problem \cite{wen2017online,vaswani2017model,wu2019factorization}. 

Similar to \cite{wen2017online}, we assume each edge $e \in \mathcal{E}$ is associated with a known feature vector $\boldsymbol{x}_e \in \mathbb{R}^d$ and an unknown coefficient vector $\boldsymbol{\theta} \in \mathbb{R}^d$, where $d$ is the dimension of the feature vector. Previous works assume that for all $e \in \mathcal{E}$, $p(e)$ is \textit{well-approximated} by $\boldsymbol{x}_e^T\boldsymbol{\theta}$. We assume malicious users can occasionally corrupt the diffusion process to mislead the agent into selecting sub-optimal seed sets. At each round $t$, if user $u$ is malicious, it can corrupt all the connected out edges with activation probabilities to its neighbors by $c_{u,t}$. Formally, the behavior of a corrupted user satisfies the following assumption. 
\begin{assumption}
\label{assump:clb-p}
For all $u,v \in \mathcal{V}$ with $e = (u,v) \in \mathcal{E}$, let $p(e)$ be the probability that user $v$ can be activated by $u$ at round $t$. For a normal user, the activation probability $p(e)$ of any out-edge $e$ can always be well-approximated as
\begin{gather}
p(e) = \boldsymbol{x}_{e}^T\boldsymbol{\theta}, \label{eq:p-normal}
\end{gather}
whereas for any corrupted user $u \in \mathcal{V}$, the activation probability of its out-edge $e$ at $t$ is given by
\begin{gather}
p_t(e) = \boldsymbol{x}_{e}^T\boldsymbol{\theta} + c_{u,t}. 
\end{gather}
\end{assumption}
In real-world applications, the activation probabilities of corrupted users' out-edges are often well-approximated by $\boldsymbol{x}_e^T\boldsymbol{\theta}$, similar to normal users; Nevertheless, a small fraction of them can be adversarially corrupted at time step $t$ with level $c_{u,t}$. Therefore, since $c_{u,t}$ can become zero at some time intervals, learning the ground truth $p(e)$ is challenging. 

Similar to \cite{wen2017online,wu2019factorization}, we assume an independent cascade diffusion model. In addition, below, we define the edge semi-bandit feedback.
\begin{definition}[Edge semi-bandit feedback]
In edge semi-bandit feedback, or edge level bandit feedback, the agent observes the influenced edge; That is, at any round $t$, the agent observes an edge $e = (u,v)$ if and only if its starting node $u$ is activated. 
\end{definition}

The performance measure for the learning algorithm is the \textit{expected regret}, which is the difference between the optimal influence under perfect knowledge and the realized influence spread by the algorithm. Since computing the optimal seed set is NP-hard even under the perfect knowledge, similar to \cite{chen2016combinatorial,wen2017online,vaswani2017model,wu2019factorization,zuo2022online}, we measure the performance of the algorithm by scaled cumulative regret defined as follows.
\begin{gather}
    R^{\alpha \gamma}(T) = T \cdot f_{D,\boldsymbol{P}}(\mathcal{S}^{opt}) - \frac{1}{\alpha \gamma}\mathbb{E}[\sum_{t=1}^T f_{D,\boldsymbol{P}}(\mathcal{S}_t)],
\end{gather}
where $\alpha \gamma \in (0,1)$. 

We assume that the feature vector $\boldsymbol{x}$ and $\boldsymbol{\theta}$ satisfy the following assumption on the bandit model.
\begin{assumption}
\label{assum:sys}
For any edge $e \in \mathcal{E}$, the feature vector $\boldsymbol{x}_e$ satisfies $\norm{\boldsymbol{x}_e} \leq 1$. The unknown coefficient feature vector $\boldsymbol{\theta}$ satisfies $\norm{\boldsymbol{\theta}} \leq \Theta$ where $\Theta$ is a constant. For the normalized feature vector, $\Theta = 1$.
\end{assumption}

To measure the level of adversarial corruptions, we define the \textit{corruption level (total corruption budget)} as $C_u = \sum_{t=1}^T |c_{u,t}|$ and $C = \max_{u \in \mathcal{V}} C_u$ is the maximum corruption level. 
\begin{remark}
The presence of corruption in activation probabilities does not affect the fulfillment of Assumptions~\ref{assump:mono} and ~\ref{assump:bound}. Assumptions~\ref{assump:mono} and ~\ref{assump:bound} are originally proposed in \cite{wang2017improving} with a multi-armed bandit framework, where each edge is linked to an activation probability, without making any assumptions on how this probability is approximated using edge features. Consequently, they remains valid irrespective of the activation probability approximation model.
\end{remark}

\begin{remark}
Our definition of corruption level is an extension of the definition in \cite{bogunovic2021stochastic,he2022nearly}. We extend the original definition from one corrupted user setting to multiple corrupted users by considering the worst-case with the maximum corruption level.
\end{remark}
\section{A Confidence Weighted Bandit Solution}
\label{sec:cwb}
In this section, we propose Confidence Weighted Influence Maximization Linear UCB (CW-IMLinUCB) algorithm, which robustly learns the activation probability over the directed graph from corrupted feedback. Algorithm~\ref{alg:cwim-lin} presents the pseudocode.
\begin{algorithm}[!ht]
\caption{CW-IMLinUCB} 
\label{alg:cwim-lin}
\begin{algorithmic}[1]
\STATE \textbf{Input}: Graph $\mathcal{G} = \{\mathcal{V}, \mathcal{E}\}$, seed set cardinality $K$, oracle \textsc{ORACLE}, edge feature vector $\boldsymbol{x}_e$, $\forall e \in \mathcal{E}$, algorithm parameters $\lambda, \sigma, \beta > 0$.
\STATE \textbf{Initialization} 
\begin{itemize}
    \item $\boldsymbol{b}_0 = \boldsymbol{0} \in \mathbb{R}^d$ and $\boldsymbol{M}_0 = I \in \mathbb{R}^{d\times d}$;
    \item $\hat{\boldsymbol{\theta}} = \boldsymbol{0} \in \mathbb{R}^d$ and $\hat{p}_0(e) = 1$, for all $e \in \mathcal{E}$;
\end{itemize}
\FOR{$t = 1,2, \cdots, T$}
\STATE Choose $\mathcal{S}_t \leftarrow \textsc{ORACLE}(\mathcal{G},K,\hat{\boldsymbol{P}}_{t-1})$ where $\hat{\boldsymbol{P}}_{t-1} = \{\hat{p}_{t-1}(e)\}_{e\in\mathcal{E}}$
\STATE Observe the edge level semi-bandit feedback $\boldsymbol{y}_t \in \mathbb{R}^m$
\FOR{$e \in \mathcal{E}$}
\IF{$e \in \tilde{\mathcal{E}}_t$} 
\STATE \algorithmiccomment{weighted regression} \\
$\omega_{e,t} = \min \{1, \lambda/\norm{\boldsymbol{x}_e}_{\boldsymbol{M}_{t-1}^{-1}}\}$ 
\STATE $\boldsymbol{b}_{t} \leftarrow \boldsymbol{b}_{t-1} + \omega_{e,t}\boldsymbol{x}_ey_t(e)$
\STATE $\boldsymbol{M}_{t} \leftarrow \boldsymbol{M}_{t-1} + \sigma^{-2}\omega_{e,t}\boldsymbol{x}_e\boldsymbol{x}_e^T$
\ELSE
\STATE $\boldsymbol{b}_{t} \leftarrow \boldsymbol{b}_{t-1}$
\STATE $\boldsymbol{M}_{t} \leftarrow \boldsymbol{M}_{t-1}$
\ENDIF
\ENDFOR
\STATE $\hat{\boldsymbol{\theta}}_{t} \leftarrow \sigma^{-2}\boldsymbol{M}_{t}^{-1}\boldsymbol{b}_{t}$
\STATE $\hat{p}_t(e) = \mathbb{P}_{[0,1]}(\hat{\boldsymbol{\theta}}_{t}^T\boldsymbol{x}_e + \beta\norm{\boldsymbol{x}_e}_{\boldsymbol{M}_{t}^{-1}})$, for all $e \in \mathcal{E}$
\ENDFOR
\end{algorithmic}
\end{algorithm}

The inputs of CW-IMLinUCB are the network topology $\mathcal{G}$, the seed set cardinality $K$, the optimization algorithm \textsc{ORACLE}, the feature vectors $\boldsymbol{x}_e \in \mathbb{R}^d$, $\forall e \in \mathcal{E}$ and three algorithm hyper-parameters $\lambda, \sigma, \beta > 0$. The value of $\sigma$ is proportional to the noise in the observations and hence controls the learning rate \cite{vaswani2017model}. For each time step $t$, we define the Gram matrix $\boldsymbol{M}_{t} \in \mathbb{R}^{d\times d}$ and $\boldsymbol{b}_{t} \in \mathbb{R}^d$ as the vector summarizing the past propagations. Besides, $\hat{\boldsymbol{\theta}}_{t}$ refers to the estimation of the unknown coefficient vector at time step $t$. $\boldsymbol{M}_{t}$ and $\boldsymbol{b}_{t}$ are sufficient statistics to compute $\hat{\boldsymbol{\theta}}_t$ and estimate the activation probability $p(e)$. The parameter $\beta$ is utilized in forming the upper confidence bound (UCB) to consider the tradeoff between mean and variance, thus controls the \textit{degree of optimism} of the algorithm \cite{vaswani2017model}.

At each time step $t$, CW-IMLinUCB firstly uses the estimated UCB of the activation probability from last time step to compute the seed set $\mathcal{S}_t$ based on the given optimization algorithm $\textsc{ORACLE}$ (Line 4). Then the algorithm receives the edge semi-bandit feedback. $\tilde{\mathcal{E}}_t$ refers to the set including all the observed edges at time step $t$ and $\boldsymbol{y}_t$ is an $m$-dimensional vector with $y_t(e_i) = y_t((u,v)) = \mathbbm{1}\{v \text{ is activated via edge $e_i$ at time step } t\}$, $\forall i \in \{1,2,\ldots,m\}$, which records the activation result. Afterwards, it updates $\boldsymbol{M}$, $\boldsymbol{b}$ and the UCB of the activation probability for each edge, where $\mathbb{P}_{[0,1]}(\cdot)$ denotes the Euclidean projection onto the nearest point in the interval $[0,1]$. The algorithm utilizes the updated activation probability estimation in the seed set selection of the next round. 

Specially, different from previous works \cite{vaswani2015influence,vaswani2017model,wu2019factorization}, which directly apply the classical OFUL algorithm with ridge regression \cite{abbasi2011improved} to estimate the unknown feature vector, our algorithm assigns each edge $e$ a weight factor $\omega_{e,t}$. More precisely, the previous works estimate $\boldsymbol{\theta}$ by online ridge regression over all past observations, i.e.,
\begin{gather*}
    \boldsymbol{\theta}_{t} \rightarrow \arg \min_{\boldsymbol{\theta} \in \mathbb{R}^d}\norm{\boldsymbol{\theta}}_2^2 + \sum_{\tau=1}^{t}\sum_{e \in \tilde{\mathcal{E}}_{\tau}}\sigma^{-2}(\boldsymbol{\theta}^T\boldsymbol{x}_e - y_{\tau}(e))^2. 
\end{gather*}
However, in the presence of corruption, the previous algorithms that rely on the upper confidence bound parameter $\beta$ without accounting for corruption \cite{abbasi2011improved,wen2017online} will experience a deterioration in regret performance,  which will lead to a term $O(C\sqrt{T})$ in the regret, i.e., the regret bound is $C$ times worse than the regret without corruption \cite{he2022nearly}.


To overcome this difficulty, inspired by \cite{he2022nearly}, we use the weighted ridge regression to estimate $\boldsymbol{\theta}$ as
\begin{gather*}
    \hat{\boldsymbol{\theta}}_{t} \rightarrow \arg \min_{\boldsymbol{\theta} \in \mathbb{R}^d}\norm{\boldsymbol{\theta}}_2^2 + \sum_{\tau=1}^{t}\sum_{e \in \tilde{\mathcal{E}}_{\tau}}\omega_{e,\tau}\sigma^{-2}(\boldsymbol{\theta}^T\boldsymbol{x}_e - y_{\tau}(e))^2, 
\end{gather*}
where its closed-form solution is $\hat{\boldsymbol{\theta}}_{t} = \sigma^{-2} \boldsymbol{M}_{t}^{-1}\boldsymbol{b}_{t}$, where $\boldsymbol{M}_{t} = I + 
\sigma^{-2}\sum_{\tau=1}^{t}\sum_{e \in \tilde{\mathcal{E}}_{\tau}}\omega_{e,\tau}\boldsymbol{x}_e\boldsymbol{x}_e^T$ and $\boldsymbol{b}_{t} = \sum_{\tau=1}^{t}\sum_{e \in \tilde{\mathcal{E}}_{\tau}}\omega_{e,\tau}\boldsymbol{x}_ey_{\tau}(e)$ (line 6 to line 16). We set the weight of sample at round $t$ as $\omega_{e,t} = \min\{1,\frac{\lambda}{\norm{\boldsymbol{x}_e}_{\boldsymbol{M}_{t-1}^{-1}}}\}$, where $\lambda > 0$ is a threshold coefficient to be determined later. 

\begin{remark}
The term $\norm{\boldsymbol{x}_e}_{\boldsymbol{M}_{t-1}^{-1}}$ in line 8 of Algorithm~\ref{alg:cwim-lin} refers to the confidence radius. If $\norm{\boldsymbol{x}_e}_{\boldsymbol{M}_{t-1}^{-1}}$ is large, CW-IMLinUCB assigns a small weight $\omega_{e,t}$ to avoid the potentially large regret caused by noise and adversarial corruption, while when $\norm{\boldsymbol{x}_e}_{\boldsymbol{M}_{t-1}^{-1}}$ is small, it assigns a large weight $\omega_{e,t}$ (no more than 1) \cite{he2022nearly,wang2023online}. Therefore, with carefully selected $\lambda$, our CW-IMLinUCB algorithm can get rid of the $O(C\sqrt{T})$ term caused by the corruption in the final regret compared to the OIM algorithm with original ridge regression \cite{vaswani2015influence}. 
\end{remark}
\begin{remark}
CW-IMLinUCB has a storage complexity independent of $t$ since only $\boldsymbol{M}_t$ and $\boldsymbol{b}_t$ need to be stored and updated. We need to emphasise that CW-IMLinUCB's  computational efficiency replies heavily on the computational efficiency of \textsc{ORACLE}. Specifically, at each time step $t$, the computational complexity of CW-IMLinUCB from line 5 to line 17 is $O(md^2)$. Although updates of CW-IMLinUCB in line 8 and 16 involve matrix inversions, the operations do not incur high computational complexity since $\boldsymbol{M}_t$ only have sizes $d \times d$. 
\end{remark}

\section{Theoretical Analysis}\label{sec:ta}
In this section, we derive a regret bound of CW-IMLinUCB under Assumption~\ref{assump:mono}-\ref{assum:sys}. Notice that Assumption~\ref{assump:mono}-\ref{assump:clb-p} are standard for bandit analysis, and Assumption~\ref{assum:sys} can be satisfied by rescaling the feature vectors. 

First we introduce the following definition. 
\begin{definition}
\label{def:Estar}
Assume that the graph $\mathcal{G} = (\mathcal{V},\mathcal{E})$ includes $l$ disconnected subgraphs $\mathcal{G}_1 = (\mathcal{V}_1,\mathcal{E}_1)$, $\ldots$, $\mathcal{G}_l = (\mathcal{V}_l,\mathcal{E}_l)$, which are in a descending order according to the number of nodes of each graph. $E^*$ is defined as the number of the edges containing in the first $\min \{l,K\}$ subgraphs \cite{wen2017online},
\begin{gather}
    E^* = \sum_{i=1}^{\min\{l,K\}}|\mathcal{E}_i|,
\end{gather}
and it is easy to obtain $E^* \leq |\mathcal{E}| = m$. 

Furthermore, we introduce $\mathcal{D}_u(\mathcal{G}_i)$ as the set containing all descendants of node $u$ within graph $\mathcal{G}_i$. Let $\mathcal{P}_{u,v \in \mathcal{D}_u(\mathcal{G}_i)}$ denote the set containing all paths from node $u$ to its descendant $v \in \mathcal{D}_u(\mathcal{G}_i)$, and $\mathcal{E}_{u,d \in \mathcal{D}_u(\mathcal{G}_i)}$ denote as the set collecting all edges within $\mathcal{P}_{u,v \in \mathcal{D}_u(\mathcal{G}_i)}$. In other words, $\mathcal{E}_{u,d \in \mathcal{D}_u(\mathcal{G}_i)}$ captures the edges forming paths to all descendants of node $u$ within $\mathcal{G}_i$. For simplification, we call these edges as descendant edges. $E^c$ is defined as 
\begin{gather}
E^c = \sum_{i=1}^{\min \{l,K\}} \max_u |\mathcal{E}_{u,v \in \mathcal{D}_u(\mathcal{G}_i)}|. 
\end{gather}
In words, $E^c$ is the summation of maximum count of descendant edges within the first $\min \{l,K\}$ subgraphs, with $E^c \leq E^* \leq m$.
\end{definition}

The following lemma defines the upper confidence bound parameter $\beta$.
\begin{lemma}
\label{lem:conf-b}
For any $0 < \delta < 1$ and corruption budget $C \geq 0$, set the confidence radius $\beta = \sigma^{-2}\sqrt{d\log(1+\frac{E^*T}{d}) + 2\log (\frac{1}{\delta})} + \sigma^{-2}\lambda E^cC + \Theta$ then with probability at least $1-\delta$, for every round $t$, the good event $\xi_{t-1} = \Big \{|\boldsymbol{x}_e^T(\hat{\boldsymbol{\theta}}_{\tau-1} - \boldsymbol{\theta})| \leq \beta \sqrt{\boldsymbol{x}_e^T\boldsymbol{M}_{\tau-1}^{-1}\boldsymbol{x}_e},\forall e \in \mathcal{E}, \forall \tau \leq t \Big\}$ happens $\forall t \in \{1,2,\ldots, T\}$. $E^c \leq E^* \leq |\mathcal{E}|$ and the corresponding definitions are stated in Definition~\ref{def:Estar}. 
\end{lemma}
\begin{proof}
    See Appendix~\ref{app:proof-lem}. 
\end{proof}

The following theorem states the regret bound of CW-IMLinUCB. 
\begin{theorem}
\label{the:rgt}
Assume that the activation probability of users satisfy Assumption~\ref{assump:clb-p}. Besides, \textsc{ORACLE} is an $(\alpha,\gamma)$-approximation algorithm. Let $\Theta$ be the known upper bound on $\norm{\boldsymbol{\theta}}$, $C \geq 0$ is the corruption budget, $\lambda= \frac{\sqrt{d}}{CE^c}$. For any $\sigma > 0$ and any $x_e \in \mathbb{R}^d$, $\forall e \in \mathcal{E}$, if $\beta$ satisfies
\begin{align}
    \beta \geq \sigma^{-2}  \sqrt{d\log(1+\frac{E^*T}{d}) + 2\log (nT)} + \sigma^{-2}\lambda E^c C + \Theta, \label{eq:beta}
\end{align}
the $\alpha \gamma$-scaled regret is upper bounded as
\begin{gather}
    R^{\alpha \gamma}(T) \leq  O(dBE^*\sqrt{T}\log(nT) + BE^*E^cCd\log(nT)). \label{eq:regret}
\end{gather}
\end{theorem}
\begin{proof}
See Appendix~\ref{app:proof-the}.
\end{proof}
\begin{remark}
If we consider the individual corruption budget $C_u$, by selecting $\lambda = \frac{\sqrt{d}}{\sum_{i=1}^{\min \{l,K\}} \max_u |\mathcal{E}_{u,v \in \mathcal{D}_u(\mathcal{G}_i)}| C_u}$, we can derive the tighter regret bound $O(dE^*\sqrt{T}\log (nT) + (\sum_{i=1}^{\min\{l.K\}}\max_u |\mathcal{E}_{u,v \in \mathcal{D}_u(\mathcal{G}_i)}|) C_u E^*d\log (nT))$. Definition~\ref{def:Estar} defines $\mathcal{E}_{u,v \in \mathcal{D}_u(\mathcal{G}_i)}$ and $\mathcal{D}_u(\mathcal{G}_i)$. 
\end{remark}
\begin{remark}
The regret bound in \eqref{eq:regret} is a topology-dependent bound. By Definition~\ref{def:Estar}, $E^c$ and $E^*$ are less than $m$. Thus, the worst-case upper bound of the scaled regret yields $O(dmB\sqrt{T}\log(nT) + Bm^2Cd\log(nT))$. 
\end{remark}
\begin{remark}
When $C = 0$, i.e., no user is malicious, our setting reduces to the classic OIM problem. For unknown $C$, if a (potentially imprecise) estimation of it, namely, $\bar{C}$, is available, by selecting $\lambda= \sqrt{d}/(\bar{C}E^c)$, $\beta \geq \sigma^{-2}\sqrt{d\log(1+\frac{E^*T}{d}) + 2\log (nT)} + \sigma^{-2}\lambda E^c\bar{C} + \Theta$, we can distinguish the following cases:
\begin{itemize} 
\item If $C \leq \bar{C}$, the scaled regret is bounded by $O(dBE^*\sqrt{T}\log(nT) + BE^*E^c\bar{C}d\log(nT))$. 
\item If $C \geq \bar{C}$, the algorithm has a linear regret bound with respect to the time horizon, i.e., $O(T)$.
\end{itemize}
In addition, if we set $\overline{C} = \sqrt{T}$, then when $0 < C \leq \sqrt{T}$, the regret is upper bounded by $O(dBE^*E^c\log(nT))$.
\end{remark}
\section{Experiment}\label{sec:experiments}
Unlike the previous work dealing with corruption concerning a single agent, our work delves into a unique aspect where the impact of corruption can propagate throughout the entire network in influence maximization. In this case, the number of corrupted users is not the sole determinant of the regret bound. Indeed, the placement of corrupted nodes within the network significantly influences the extent of regret. When a corrupted node occupies a pivotal position within the network, the resulting corruption effect can surpass that caused by numerous randomly selected nodes. In other words, even a few number of corrupted users have the potential to disseminate the corruption effects across the entire network. In this section, we first highlight the importance of the position of the corrupted nodes in the network and compare the performance of our algorithm against benchmarks with a toy example. We evaluate CW-IMLinUCB on a carefully selected network topology (Figure~\ref{fig:e1-network}) and validate our algorithm's performance under the dissemination of the corruption effects. Similar to previous work \cite{vaswani2015influence,vaswani2017model,kong2023online}, we evaluate the performance of our algorithm using a randomly-generated synthetic- and real-world datasets. Besides, we compare the results to the following state-of-the-art bandit algorithms:
\begin{itemize}
    \item \textbf{$\epsilon$-greedy}: This algorithm learns the activation probability of each edge independently and uses $\epsilon$-greedy \cite{auer2002finite} to balance exploitation and exploration.
    \item \textbf{CUCB} \cite{wang2017improving}: This algorithm learns the activation probability of each edge independently via a multi-armed bandit framework. 
    \item \textbf{IMLinUCB} \cite{wen2017online}: This algorithm learns under an edge-level bandit feedback and approximates the activation probability as the inner product of known edge features and one shared unknown feature among all the edges. 
    \item \textbf{DILinUCB} \cite{vaswani2017model}: This algorithm is model-dependent and approximates the activation probability as the inner product of the known target feature vector of the edge's ending node and weight vector of the edge's source node. 
    \item \textbf{OIMLinUCB}: This algorithm is the variation of IMFB \cite{wu2019factorization} and we assume the susceptibility vector $\boldsymbol{x}$ is known in advance, which is similar to DILinUCB and IMLinUCB.
\end{itemize}
Specially, for \textbf{DILinUCB} and \textbf{OIMLinUCB}, we also implement their variants with confidence weighted regression (\textbf{CW-DILinUCB} and \textbf{CW-OIMLinUCB}) to compare. Additionally, for all implemented algorithms, the DegreeDiscount algorithm \cite{chen2009efficient} is used as the \textsc{ORACLE}. ALL the experimental results are the average of ten independent runs. In all plots, error bars indicate the standard deviations divided by $\sqrt{10}$.
\subsection{Toy Example}\label{sec:e0}
In the toy example experiment, we implement our algorithm in a ten-node network with a single corrupted user and select one seed. The network is an Erdős-Rényi graph and creates possible edges with probability 0.3. Figure~\ref{fig:e1-network} shows the network structure. We consider a variation of flip-$\boldsymbol{\theta}$ attack as the corrupted behavior of the corrupted users. Flip-$\theta$ attack simply flips the reward from $\boldsymbol{\theta}^T\boldsymbol{x}$ to $-\boldsymbol{\theta}^T\boldsymbol{x}$ \cite{bogunovic2021stochastic,he2022nearly}. Considering activation probability acts as reward in online influence maximization, we add one constant to the reward and then make a flip-$\theta$ attack in our experimental setting. Thus, corrupted users trick the learning algorithm by changing the activation probability to become $p(e) = \max(0,0.05 - \boldsymbol{x}_e^T\boldsymbol{\theta})$ for the first $C_T = 100$ rounds. In the remaining rounds, the corrupted user acts normally. The activation probabilities in normal manner follow \eqref{eq:p-normal} in Assumption~\ref{assump:clb-p}. Each dimension of feature vectors $\boldsymbol{x}_e \in \mathbb{R}^{25}$, $\forall e \in \mathcal{E}$ and $\boldsymbol{\theta} \in \mathbb{R}^{25}$ is generated randomly from uniform distribution $U(0,0.1)$ and then feature vectors are normalized. The average activation probability over edges is 0.175. 
\begin{figure}
    \centering
    \includegraphics[width = 0.5\linewidth]{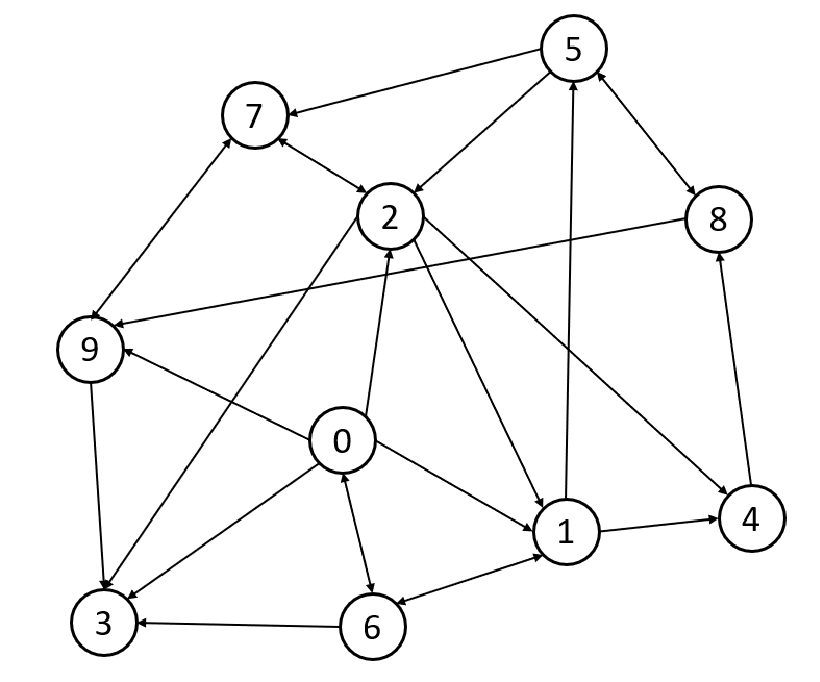}
    \caption{Network structure of Experiment I.}
    \label{fig:e1-network}
    \vspace{-10pt}
\end{figure}

In Figure~\ref{fig:e1-network}, Node 3 has no out-edges, whereas Node 1 has three in-edges and two out-edges. Therefore, the latter, if corrupted, can disseminate the corruption effect. Figure~\ref{fig:e0} shows the cumulative regret, where $\mathcal{U}=\{3\}$ and $\mathcal{U} = \{1\}$ in the legend indicate the corruption of Node 3 and Node 1, respectively.
\begin{figure}[!ht]
\centering
 \centering
  \includegraphics[width=.6\linewidth]{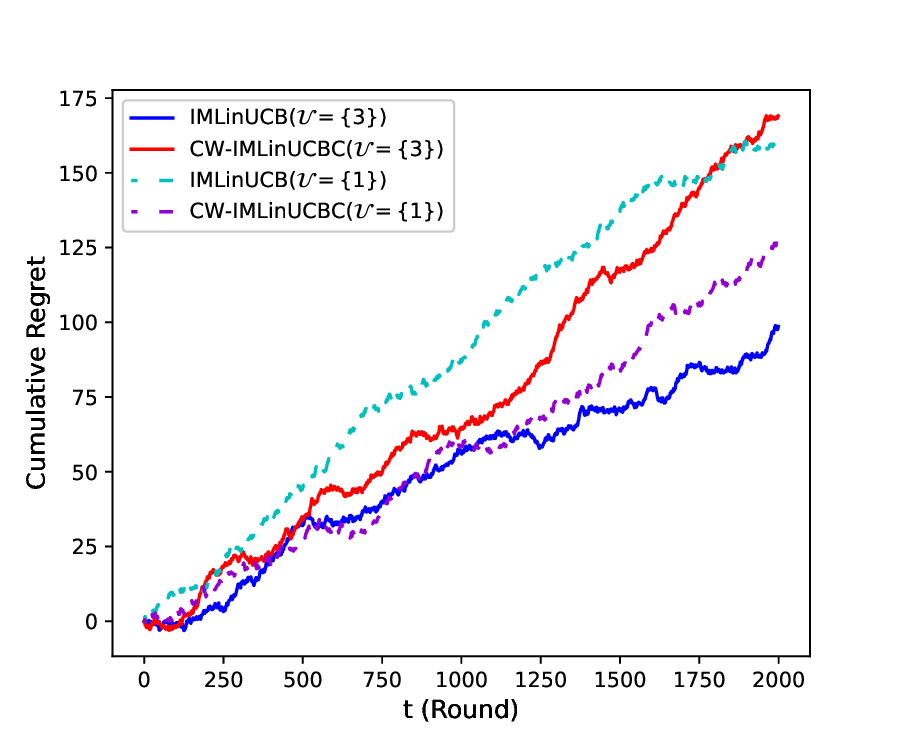} 
  \caption{Effects of various corrupted user positions ($n = 10$).}
\label{fig:e0}
\end{figure}

According to Figure~\ref{fig:e0}, when Node 3 is malicious, IMLinUCB performs better than CW-IMLinUCB. This is based on the fact that Node 3 cannot disseminate the corruption effect due to the network structure. Indeed, in this case, the extra term in the UCB of our corruption-robust algorithm adds an unnecessary exploration that degrades the performance w.r.t. the corruption-agnostic algorithm. In contrast, if Node 1 is corrupted, that term is vital, and our algorithm has a superior performance. More precisely, if Node $3$ is malicious, the additional regret of our algorithm stems from the overestimation of corruption within the upper confidence bound. In $\beta$, the second term $\sigma^{-2}\lambda E^c C$ denotes the potential spread of corruption throughout the whole network with maximal corruption budget to account for the uncertainty. However, such a scenario occurs only when the crucial nodes in the most vital positions act in a corrupted manner permanently. Thus, the over-estimated corruption in $\beta$ will increase exploration in the learning process. We emphasize that the selection of $\beta$ rests on the assumption that the positions of corrupted users remain unknown. This assumption aligns with real-world scenarios with hidden corrupted users whose exact positions cannot be identified. By assuming unknown positions, the term $\sigma^{-2}\lambda E^c C$ in $\beta$ remains indispensable.
\subsection{Experiments on Synthetic Dataset}
\label{sec:e1}
As described in Section~\ref{sec:problem}, the positions of the corrupted users play a crucial role. To demonstrate this, we first implement our synthetic dataset on the network with $n = 50$ and $K =2$. 

The network is an Erdős-Rényi graph and creates possible edges with probability 0.3. It has in total $m = 687$ edges. The corrupted users trick the learning algorithm by changing the activation probability to become $p(e) = \max(0,0.05 - \boldsymbol{x}_e^T\boldsymbol{\theta})$ for the first $C_T = 200$ rounds similar to the toy example. In the remaining rounds, the corrupted user acts normally. The activation probabilities in normal manner follow \eqref{eq:p-normal} in Assumption~\ref{assump:clb-p}. The generation of feature vectors $\boldsymbol{x}_e \in \mathbb{R}^{25}$, $\forall e \in \mathcal{E}$ and $\boldsymbol{\theta} \in \mathbb{R}^{25}$ follow the same setting as toy example. The average activation probability over edges is 0.0295. 

We first apply IMLinUCB and CW-IMLinUCB algorithms to two different corrupted user sets, namely, $\mathcal{U}=\{0,37\}$, selected at random, and $\mathcal{U}=\{34,36\}$, both of which are directly connected to the optimal seed set $\{7,43\}$ according to \textsc{ORACLE}.

\begin{figure}[!ht]
\centering
  \includegraphics[width=.6\linewidth]{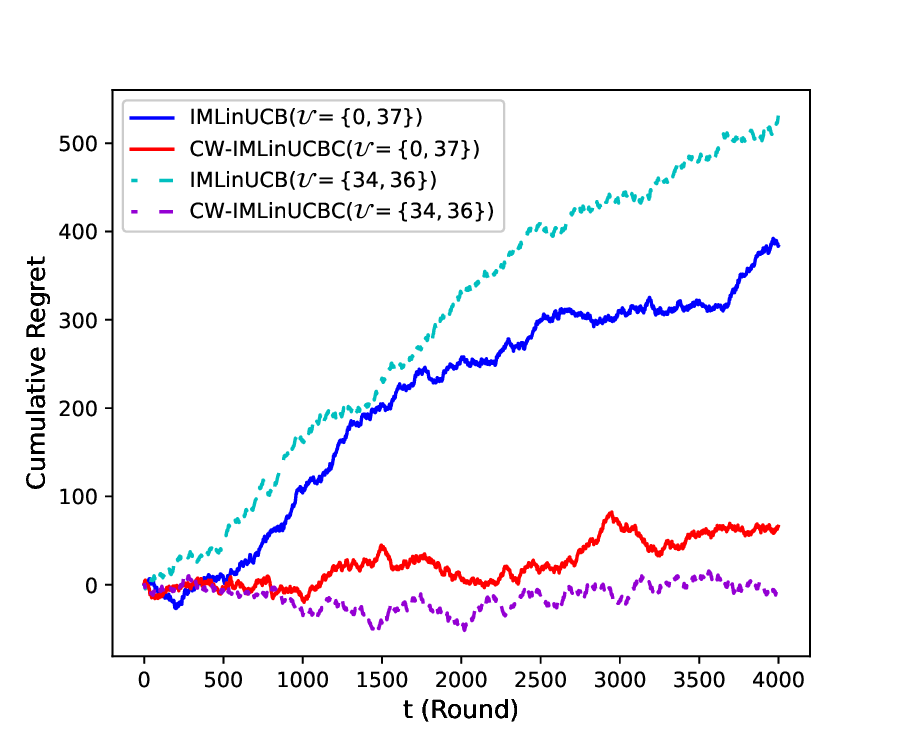}
\caption{Effects of various corrupted user positions ($n = 50$).}
\label{fig:e1.1}
\end{figure}
Figure~\ref{fig:e1.1} shows the cumulative regret with different corrupted users. Intuitively, the positions of Nodes 34 and 36 are more crucial than Nodes 0 and 37, so the corruption in those positions would cause more adverse effects. In Figure~\ref{fig:e1.1}, the performance of IMLinUCB verifies this point. When Nodes 34 and 36 are malicious, the regret of IMLinUCB is higher, and the selected seed set activates fewer nodes in the network compared to the situation when Nodes 34 and 36 are malicious. CW-IMLinUCB algorithm can overcome this difficulty. Thus, it outperforms IMLinUCB in both cases. In addition, in some rounds, the cumulative regret is below zero. That is because the algorithm calculates regret by performing the independent cascading process for both the optimal seed set and the seed set selected by our algorithm, then comparing the number of activated nodes at each time step. Due to the properties of independent cascading, the number of activated nodes can vary from round to round. Consequently, there is a possibility that the seed set selected by our algorithm might activate more nodes than the optimal seed set.

We also evaluate the performance of our proposed algorithm compared to the state-of-the-art algorithms under the same setting in previous experiment ($K = 2$, Nodes 0 and 37 are corrupted users). To guarantee that the activation probability is exactly $p(e = (u,v)) = \boldsymbol{x}_e^T\boldsymbol{\theta} = \boldsymbol{x}_v^T\boldsymbol{\theta}_u$ for all the implemented algorithms, we follow the setting in \cite{wu2019factorization}: We first randomly sample $\boldsymbol{x}_v \in \mathbb{R}^{d_1}$ and $\boldsymbol{\theta}_u \in \mathbb{R}^{d_1}$ for DILinUCB and OIMLinUCB algorithms. Then, for each edge $e = (u,v) \in \mathcal{E}$, we take the outer product on $\boldsymbol{x}_v$ and $\boldsymbol{\theta}_u$ and reshape it into a $d = d_1\times d_1$-dimensional vector, which is the edge feature vector $\boldsymbol{x}_e$ in the IMLinUCB algorithm with $d_1 = d_2 = 5$. Therefore, the IMLinUCB only needs to recognize the diagonal terms in the outer product. Figure~\ref{fig:e1.2-rgt} and \ref{fig:e1.2-reward} show the cumulative regret and the average reward, respectively. 
\begin{figure}[!ht]
\centering
\begin{subfigure}{.275\textwidth}
 \centering
  \includegraphics[width=.99\linewidth]{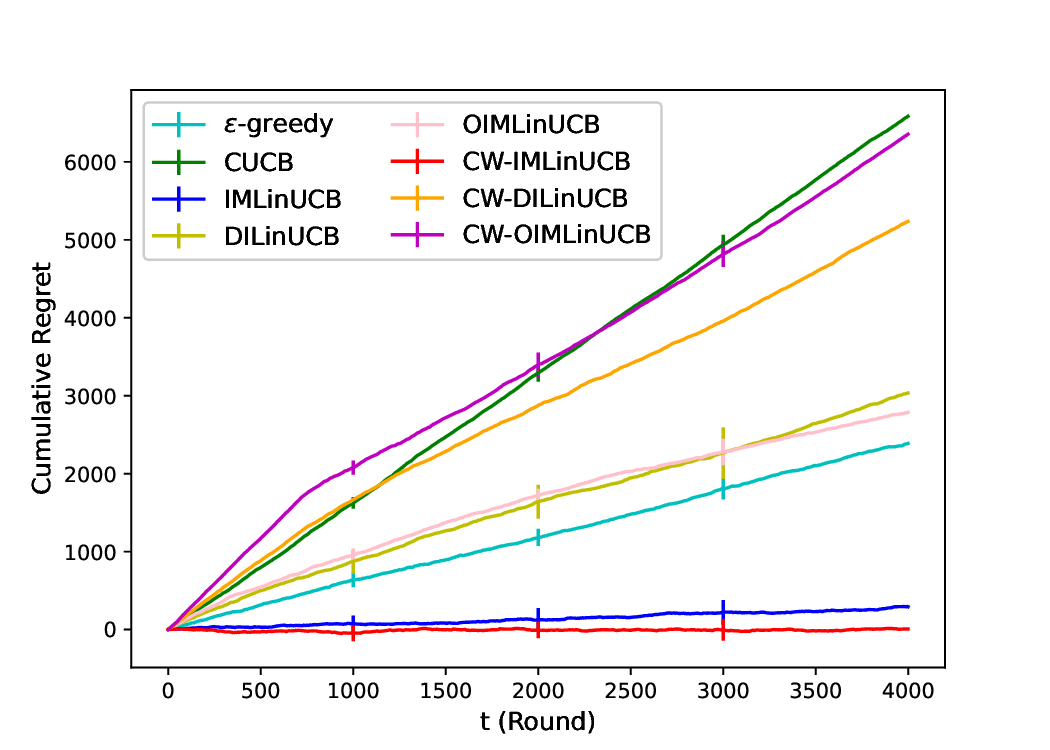}
  \caption{Cumulative regret.}
  \label{fig:e1.2-rgt}
\end{subfigure}
\begin{subfigure}{.195\textwidth}
  \centering
  \includegraphics[width=.99\linewidth]{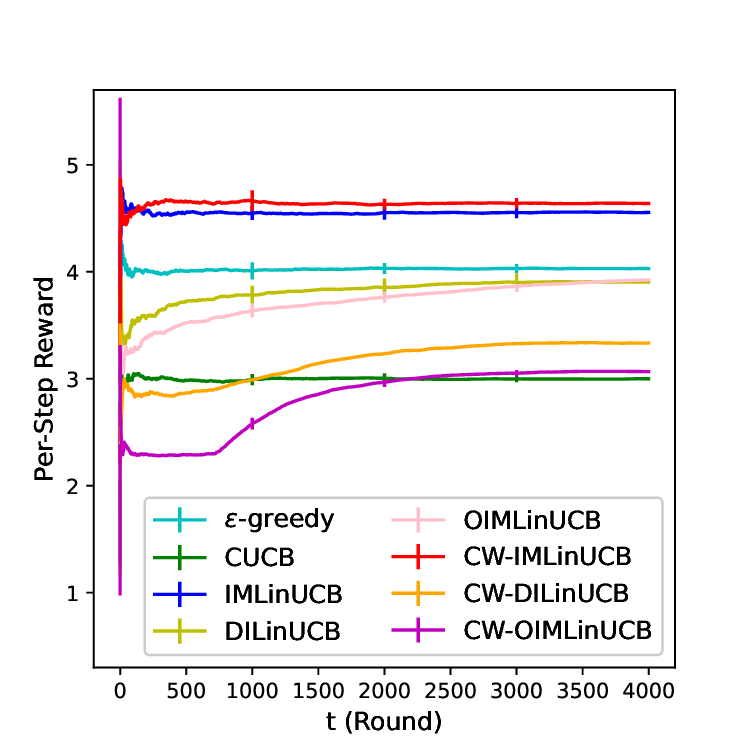} 
  \caption{Per-step reward.}
  \label{fig:e1.2-reward}
\end{subfigure}
\caption{Result of Experiment I.}
\label{fig:e1.2}
\end{figure}

Figure~\ref{fig:e1.2} shows that our proposed algorithm has the lowest regret and the highest reward compared to all other methods. The algorithms CW-DILinUCB and CW-OIMLinUCB have higher regret compared to DILinUCB and OIMLinUCB. Although CW-DILinUCB and CW-OIMLinUCB integrate the weighted regression into the algorithm, their performance is not as good as expected. 

\subsection{Experiments on Real World Dataset}
We implement our algorithm in a subgraph of the Facebook network data \cite{snapnets}. The dataset has 4039 nodes and 88234 edges, and the subgraph includes the first $n = 300$ nodes and $|\mathcal{E}| = 2046$ edges. The average edge activation probability on this subgraph is 0.0497. Several authors use the Facebook dataset to evaluate the effectiveness of OIM algorithms \cite{vaswani2015influence,vaswani2017model}. A challenge is the absence of information about edge activation probabilities, seed set sizes, and other parameters. To address that, we use the same approach in previous work \cite{vaswani2015influence}, treating the sampled values as the ground truth. The data sampling process is as follows.

We first choose $K = 20$ for the seed set. There are $n_c = 20$ randomly-selected corrupted users in the network. The generation of feature vectors is the same as Experiment I in Section~\ref{sec:e1}. The corrupted users also follow the previous setting with $C_T = 1000$. Figure~\ref{fig:e2-rgt} and \ref{fig:e2-reward} respectively show the cumulative regret and average reward, i.e., the number of per-step activated users. 
\begin{figure}[!ht]
\centering
\begin{subfigure}{.275\textwidth}
 \centering
  \includegraphics[width=.99\linewidth]{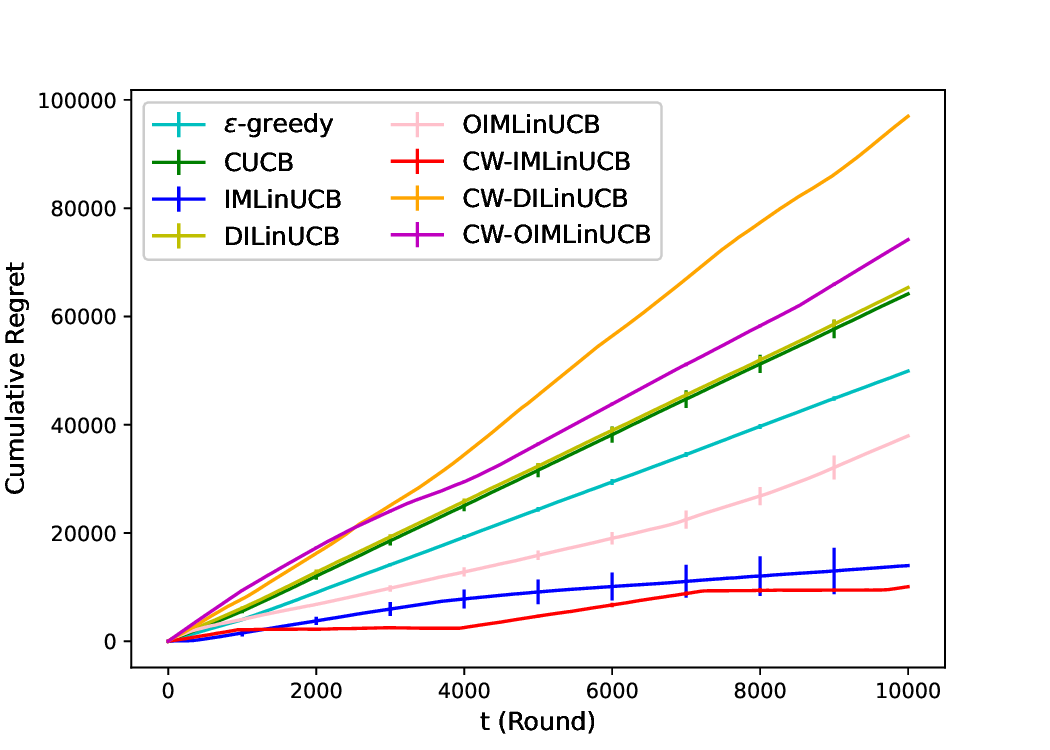}
  \caption{Cumulative regret.}
  \label{fig:e2-rgt}
\end{subfigure}
\begin{subfigure}{.195\textwidth}
  \centering
  \includegraphics[width=.99\linewidth]{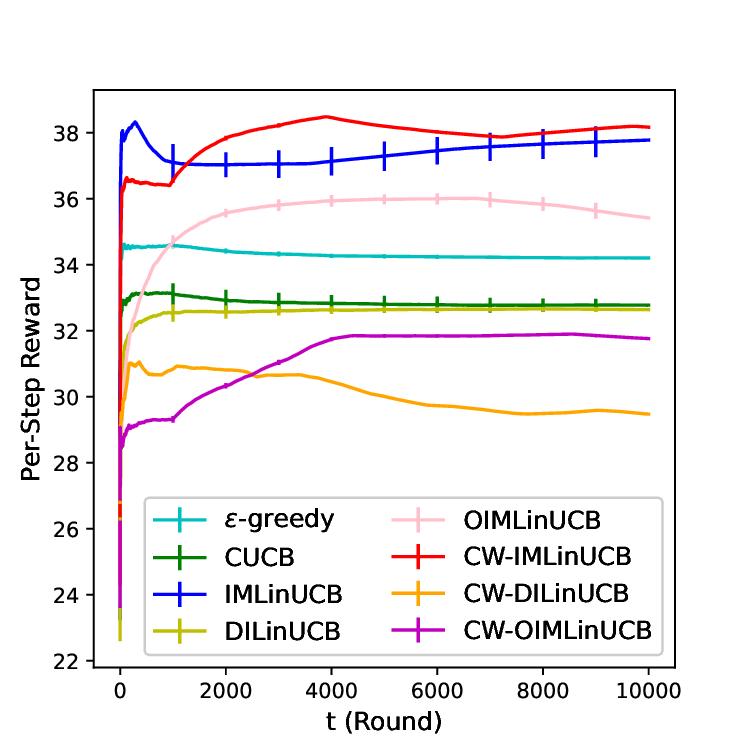} 
  \caption{Per-step reward.}
  \label{fig:e2-reward}
\end{subfigure}
\caption{Result of Experiment II.}
\label{fig:e2}
\end{figure}

Moreover, we evaluate the performance of our proposed algorithm under different corruption levels. Similar to previous numerical simulations, the confidence weighted regression is not compatible to DILinUCB and OIMLinUCB algorithms, hence we omit the plot of CW-DILinUCB and CW-OIMLinUCB in the following experiments. Figure~\ref{fig:e2-cor1} shows the regret of all algorithms under different time horizons of the corruption $C_T$ with a fixed set of $n_c = 20$ corrupted users. Figure~\ref{fig:e2-cor2} shows the performance of our algorithm when the number of corrupted users $n_c$ changes while $C_T = 1000$ remains fixed. Particularly, when $n_c \geq 10$, the experiments share the same ten corrupted users. For each experiment, we add the randomly selected users to the previous corrupted set of users. Both experiments have $K = 20$. 
\begin{figure}[!ht]
\centering
\begin{subfigure}{.24\textwidth}
 \centering
  \includegraphics[width=.99\linewidth]{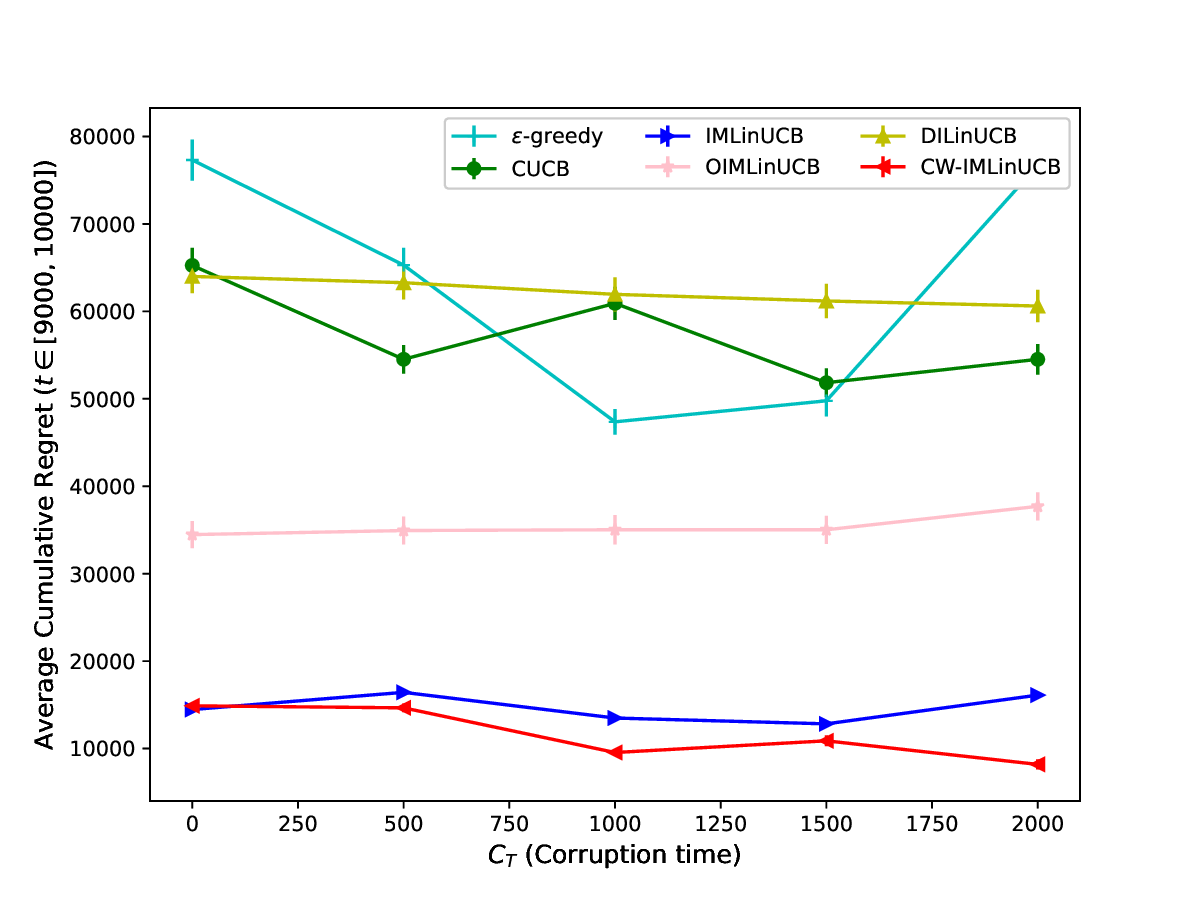}
  \caption{Regret under different corruption time.}
  \label{fig:e2-cor1}
\end{subfigure}
\begin{subfigure}{.24\textwidth}
  \centering
  \includegraphics[width=.99\linewidth]{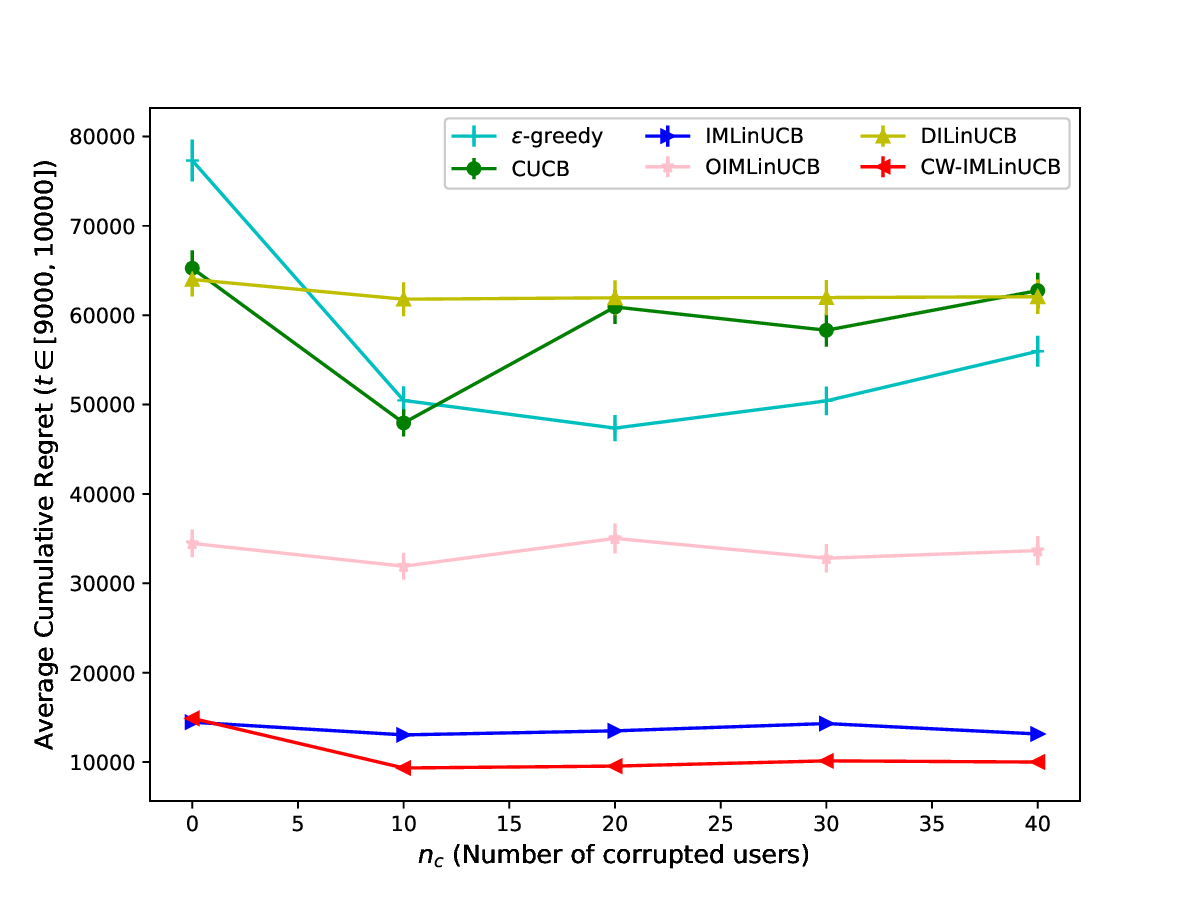} 
  \caption{Regret under different number of corrupted users}
  \label{fig:e2-cor2}
\end{subfigure}
\caption{Result of different corruption levels.}
\label{fig:e3}
\end{figure}

In addition, we evaluate the time complexity and memory requirements of our proposed algorithm. In all experiments across various networks, we generated the network using the same process as in the experiments on the synthetic dataset. We conducted all the experiments with a horizon of $T = 5000$ and the corruption time $C_T = 200$. Figure~\ref{fig:e4} shows the results.
\begin{figure}[!ht]
\centering
\begin{subfigure}{.24\textwidth}
 \centering
  \includegraphics[width=.99\linewidth]{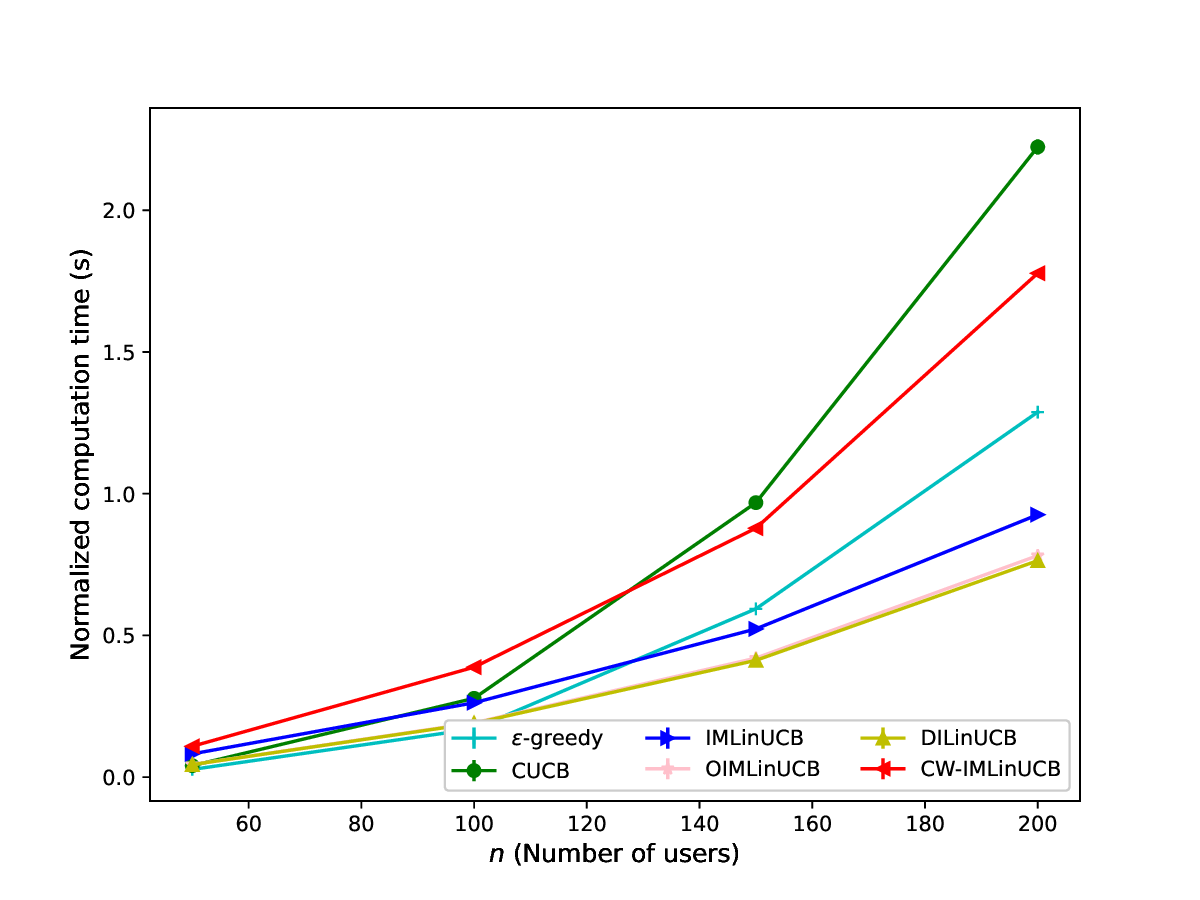}
  \caption{Normalized computation time.}
  \label{fig:e4-tc}
\end{subfigure}
\begin{subfigure}{.24\textwidth}
  \centering
  \includegraphics[width=.99\linewidth]{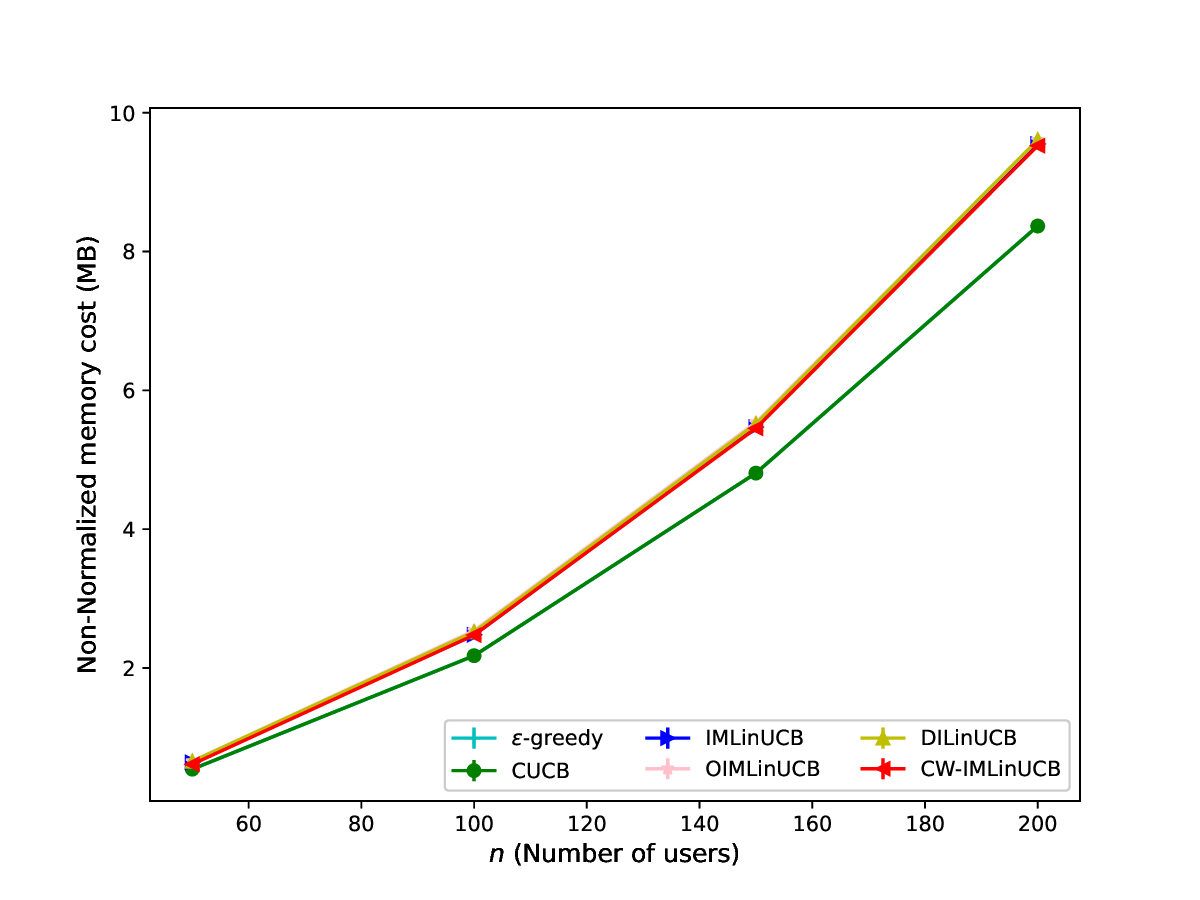} 
  \caption{Non-normalized memory cost}
  \label{fig:e4-mc}
\end{subfigure}
\caption{Complexity analysis under different networks.}
\label{fig:e4}
\end{figure}

From the experimental results, we conclude the followings:
\begin{itemize}
\item  In all experiments, our proposed algorithm, CW-IMLinUCB, outperforms other methods. 
\item As Figure~\ref{fig:e2} shows, CW-IMLinUCB has several inflection points in the regret plot. That is because the algorithm uses the weighted regression, i.e., a small weight for a large confidence radius, and vice versa. That prevents a potentially large regret caused by the corruption and consequently guarantees the superior performance of CW-IMLinUCB. Therefore, the abrupt change of the upper confidence bound caused by the weight change influences the seed set and reflects the inflection points in the plot. 
\item Compared to DILinUCB and OIMLinUCB, enhancing CW-DILinUCB and CW-OIMLinUCB with weighted ridge regression does not improve their performance compared to the vanilla version. The reason is that the weighted ridge regression framework is effective in dealing with the corruption effect under our proposed framework whereas its applicability to other frameworks in online influence maximization remains limited. Thus, our proposed structure is more compatible with IMLinUCB framework and exploits the weighted regression's strengths.
\item According to Figure~\ref{fig:e3}, when there is no corruption ($C_T = 0$ or $n_c = 0$), the performance of IMLinUCB is better than CW-IMLinUCB. The higher regret of CW-IMLinUCB here comes from the second $\sigma^{-2}\lambda E^c C$ of $\beta$. This term estimates the corruption effect within the upper confidence bound, however, in scenarios without corruption, it adds unnecessary exploration and degrades the performance of CW-IMLinUCB. 
\item According to Figure~\ref{fig:e2-cor1}, the benefits of CW-IMLinUCB become increasingly pronounced with longer corruption intervals when compared to state-of-the-art algorithms. The small fluctuation of the plot reflects the uncertainties in the diffusion process. This shows the robustness of CW-IMLinUCB algorithm in dealing with different corruption levels and its superiority compared to state-of-the-art algorithms.  
\item In Figure~\ref{fig:e2-cor2}, unlike Figure~\ref{fig:e2-cor1}, although CW-IMLinUCB has the lowest regret with different $n_c$, its regret does not significantly change among different experiments and just increased small amount with increased number of corrupted users. That is because the number of corrupted users is not the sole- or main determinant of the regret. Rather, it is the placement of corrupted nodes within the network that significantly influences the extent of regret. Merely increasing the number of corrupted nodes, without putting the corrupted node in influential positions, does not necessarily yield more corruption effect, thus the regret performance might remain steady. 
\item Compared with CW-IMLinUCB, there is no obvious trend in the regret plots of other algorithms in Figure~\ref{fig:e3}. The reason is that other algorithms do not take corruption into account, which also increases uncertainty in the behavior of their learning processes in facing corruptions.  
\item As shown in Figure~\ref{fig:e4}, the computation time and memory cost of CW-IMLinUCB increase with the network size, a trend observed across all state-of-the-art algorithms. The time cost of our proposed CW-IMLinUCB algorithm is slightly higher than that of other state-of-the-art algorithms. However, its memory cost is comparable to other algorithms with linear bandit structures, such as IMLinUCB, DILinUCB, and OIMLinUCB.
\end{itemize}
\section{Conclusion and Discussion}\label{sec:conclusion}
In this work, we study the OIM problem in presence of corrupted users. We propose an algorithm CW-IMLinUCB that integrates the weighted ridge regression into an OIM algorithm with bandit feedback. We present a theoretical guarantee for our proposed algorithm. Experiments on synthetic- and real-world datasets confirm the effectiveness and robustness of our proposed algorithm. According to our analysis, the position of the corrupted users can influence the regret bound. Currently, our regret bound only considers the worst-case that the corrupted users can easily disseminate the corruption effect across the network. In future work, one can consider to integrate corrupted user detection mechanisms into our framework and consider the corrupted user location dependent algorithm. Thus, the regret bound could be further tighter. Furthermore, developing an algorithm that competes with a dynamically corrupting algorithm is another interesting research direction. 



\appendix
\subsection{Proof of Lemma~\ref{lem:conf-b}}\label{app:proof-lem}
Here we first present an auxiliary lemma in the following.
\begin{lemma}\label{lem:inv}
Let $A, B$ be a Hermitian matrix in $\mathbb{R}^{d\times d}$ and suppose $A,B \succ 0$, then $A \succeq B$ if and only if $B^{-1} \succeq A^{-1}$ \cite{horn2012matrix}.
\end{lemma}
The proof of Lemma~\ref{lem:conf-b} is as follows.
\begin{proof}
Let $\mathcal{H}_t$ denote the history by the end of time step $t$ hence $\{\mathcal{H}_t\}_{t=0}^{\infty}$ is a filtration. Note that $\hat{p}_t$ is $\mathcal{H}_{t-1}$-adaptive and $\mathcal{S}_t$ is also $\mathcal{H}_{t-1}$-adaptive\cite{vaswani2017model}. For $t \in \{1,2,\ldots, T\}$ and $e= (u,v) \in \bar{\mathcal{E}}_t$, define
\begin{gather*}
    \eta_t(e) = y_t(e) - \boldsymbol{\theta}^T\boldsymbol{x}_e - c_{u,t}. \label{eq:eta}
\end{gather*}
Note that $\eta_t(e)$ is $\mathcal{H}_{t-1}$-adaptive. Besides, it satisfies  $\norm{\eta_t(e)} \leq 1$ and $\mathbb{E}[\eta_t(e)|\mathcal{H}_{t-1}] = 0$ \cite{vaswani2017model}, hence they are conditionally sub-Gaussian with constant $R = 1$ \cite{wen2017online}. We further define 
\begin{align*}
\boldsymbol{V}_{t} &= \sigma^2\boldsymbol{M}_{t} = \sigma^2 \boldsymbol{I} + \sum_{\tau=1}^t \sum_{e \in \tilde{\mathcal{E}}_{\tau}}\omega_{e,\tau}\boldsymbol{x}_e\boldsymbol{x}_e^T \notag \\
\boldsymbol{S}_{t} &= \sum_{\tau=1}^t\sum_{e \in \tilde{\mathcal{E}}_{\tau}} \omega_{e,\tau}\boldsymbol{x}_e\eta_{\tau}(e) \notag \\
&= \boldsymbol{b}_{t} - \sigma^2(\boldsymbol{M}_{t} - I)\boldsymbol{\theta} - \sum_{\tau=1}^t\sum_{e \in \tilde{\mathcal{E}}_{\tau}}\omega_{e,t}\boldsymbol{x}_ec_{u,t}.
\end{align*}
Thus we have 
\begin{align*}
\hat{\boldsymbol{\theta}}_{t} - \boldsymbol{\theta} = \boldsymbol{M}_{t}^{-1}(\sigma^{-2}\boldsymbol{S}_{t} + \sigma^{-2}\sum_{\tau=1}^t\sum_{e \in \tilde{\mathcal{E}}_{\tau}}\omega_{e,\tau}\boldsymbol{x}_ec_{u,\tau} - \boldsymbol{\theta}). 
\end{align*}
It implies
\begin{align}
|\boldsymbol{x}_e^T(\hat{\boldsymbol{\theta}}_{t-1} - \boldsymbol{\theta})| &\leq |\boldsymbol{x}_e\boldsymbol{M}_{t-1}^{-1}\boldsymbol{\theta}| + \sigma^{-2}|\boldsymbol{x}_e\boldsymbol{M}_{t-1}^{-1}\boldsymbol{S}_{t-1}| \notag \\
& \quad + \sigma^{-2}|\boldsymbol{x}_e\boldsymbol{M}_{t-1}^{-1}\sum_{\tau=1}^t\sum_{e \in \tilde{\mathcal{E}}_{\tau}}\omega_{e,\tau}\boldsymbol{x}_ec_{u,\tau}| \notag \\
&\leq \norm{\boldsymbol{x}_e}_{\boldsymbol{M}_{t-1}^{-1}}(\norm{\boldsymbol{\theta}}_2 + \underbrace{\sigma^{-2}\norm{\boldsymbol{S}_{t-1}}_{\boldsymbol{M}_{t-1}^{-1}}}_{\text{stochastic error}} \notag \\
& \quad + \underbrace{\sigma^{-2}\norm{\sum_{\tau=1}^{t-1}\sum_{e \in \tilde{\mathcal{E}}_{\tau}}\omega_{e,\tau}\boldsymbol{x}_ec_{u,\tau}}_{\boldsymbol{M}_{t-1}^{-1}}}_{\text{corruption error}}), \label{eq:dec} 
\end{align}
where \eqref{eq:dec} is based on Cauthy-Schwarz inequality and matrix operator inequality. 

The stochastic error can be bounded by the concentration Lemma J.2 in \cite{abbasi2011improved}: According to the Appendix I.1 in \cite{he2022nearly}, we introduce the auxiliary vectors $\tilde{\boldsymbol{x}}_{e,t} = \sqrt{\omega_{e,t}}\boldsymbol{x}_e$ and $\tilde{\eta}_t(e) = \sqrt{\omega_{e,t}}\eta_t(e)$. Then, it holds
\begin{gather}
\norm{\tilde{\boldsymbol{x}}_{e,t}}_2 \leq 1, \quad \tilde{\eta}_t(e) \text{ is R-sub Gaussian}, R \leq 1.
\end{gather}
With this notation, with probability at least $1-\delta$, it holds
\begin{align*}
\sigma^{-2}\norm{\boldsymbol{S}_{t-1}}_{\boldsymbol{M}_{t-1}^{-1}} &= \sigma^{-2}\norm{\sum_{\tau=1}^t\sum_{e \in \tilde{\mathcal{E}}_{\tau}} \omega_{e,\tau}\boldsymbol{x}_e\eta_{\tau}(e)}_{\boldsymbol{M}_{t-1}^{-1}} \notag \\
&= \sigma^{-2}\norm{\sum_{\tau=1}^{t-1}\sum_{e \in \tilde{\mathcal{E}}_{\tau}} \tilde{\boldsymbol{x}}_{e,\tau}\tilde{\eta}_{\tau}(e)}_{\boldsymbol{M}_{t-1}^{-1}} \notag \\
\leq &\sigma^{-2}\sqrt{2\log(\frac{\det(\boldsymbol{M}_{t-1})^{1/2}\det(\boldsymbol{M}_{0})^{-1/2}}{\delta})}, 
\end{align*}
where the last inequality is based on the Lemma J.2 and Theorem 1 in \cite{abbasi2011improved}. It is satisfied with probability at least $1-\delta$ with $\delta \in (0,1)$. Notice that $\det(\boldsymbol{M}_{0}) = \det(I) = 1$. Moreover, from the trace-determinant inequality, we have
\begin{align*}
\det(\boldsymbol{M}_{t-1})^{1/d} \leq \frac{\Tr(\boldsymbol{M}_{t-1})}{d} &= 1 + \frac{1}{d}\sum_{\tau=1}^{t-1}\sum_{e \in \tilde{\mathcal{E}}_{\tau}} \tilde{\boldsymbol{x}}_{e,\tau}\tilde{\boldsymbol{x}}_{e,\tau}^T \notag \\
&\leq 1 + \frac{(t-1)E^*}{d},
\end{align*}
where the last inequality follows from the assumption that $\norm{\tilde{\boldsymbol{x}}_{e,t}}_2 \leq 1$ and $|\tilde{\mathcal{E}}_t| \leq E^*$. Thus, with probability at least $1-\delta$, the stochastic error can be bounded by
\begin{gather*}
    \sigma^{-2}\norm{\boldsymbol{S}_{t-1}}_{\boldsymbol{M}_{t-1}^{-1}} \leq \sigma^{-2}\sqrt{d\log(1+\frac{tE^*}{d}) + 2\log (\frac{1}{\delta})}. 
\end{gather*}

The corruption error can be bounded by
\begin{align*}
&\quad \sigma^{-2}\norm{\sum_{\tau=1}^{t-1}\sum_{e \in \tilde{\mathcal{E}}_{\tau}}\omega_{e,\tau}\boldsymbol{x}_ec_{u,\tau}}_{\boldsymbol{M}_{t-1}} \notag \\
&\leq \sigma^{-2}\sum_{\tau=1}^{t-1}\sum_{e \in \tilde{\mathcal{E}}_{\tau}}\norm{ \boldsymbol{M}_{t-1}^{-1/2}\omega_{e,\tau}\boldsymbol{x}_ec_{u,\tau}}_2 \notag \\
&= \sigma^{-2}\sum_{\tau=1}^{t-1}\sum_{e \in \tilde{\mathcal{E}}_{\tau}} |c_{u,\tau}| \times \omega_{e,\tau} \norm{ \boldsymbol{M}_{t-1}^{-1/2}\boldsymbol{x}_e}_2 \notag \\
&\leq \sigma^{-2}\sum_{\tau=1}^{t-1}\sum_{e \in \tilde{\mathcal{E}}_{\tau}} |c_{u,\tau}| \lambda \leq \sigma^{-2}\lambda E^c C, 
\end{align*}
where the first inequality holds due to the inequality that $\norm{a+b}_2 \leq \norm{a} + \norm{b}$. The second one holds due to the definition of $\omega_{e,t}$. By the assumption $\norm{\boldsymbol{\theta}} \leq \Theta$, and after substituting the results, we arrive at
\begin{align}
&|\boldsymbol{x}_e^T(\hat{\boldsymbol{\theta}}_{t-1} - \boldsymbol{\theta})| \notag \\
&\leq \norm{\boldsymbol{x}_e}_{\boldsymbol{M}_{t-1}^{-1}}(\sigma^{-2}\sqrt{d\log(1+\frac{E^*T}{d}) + 2\log (\frac{1}{\delta})} \notag \\
&\quad \quad + \sigma^{-2}\lambda E^cC + \Theta)
\end{align}
\end{proof}

\subsection{Proof of Theorem~\ref{the:rgt}}\label{app:proof-the}
\begin{proof}
The scaled regret at time $t$ is $R^{\alpha \gamma}_t = f_{D,\boldsymbol{P}}(\mathcal{S}^{opt}) - \frac{1}{\alpha \gamma} f_{D,\boldsymbol{P}}(\mathcal{S}_t)$. Using the naive bound $R^{\alpha \gamma}_t \leq n -K$, the scaled cumulative regret can be decomposed into 
\begin{align*}
    R^{\alpha \gamma}(T) 
    &= \sum_{t=1}^T\mathbb{P}(\xi_{t-1})\mathbb{E}\Bigl\{[f_{D,\boldsymbol{P}}(\mathcal{S}^{opt}) - \frac{1}{\alpha \gamma} f_{D,\boldsymbol{P}}(\mathcal{S}_t)]\Big|\xi_{t-1} \Bigr\} \notag \\
    & + \sum_{t=1}^T\mathbb{P}(\bar{\xi}_{t-1})\mathbb{E}\Bigl\{[f_{D,\boldsymbol{P}}(\mathcal{S}^{opt}) - \frac{1}{\alpha \gamma} f_{D,\boldsymbol{P}}(\mathcal{S}_t)]\Big|\bar{\xi}_{t-1} \Bigr\} \notag \\
    &\leq \sum_{t=1}^T\mathbb{E}\Bigl\{[f_{D,\boldsymbol{P}}(\mathcal{S}^{opt}) - \frac{1}{\alpha \gamma} f_{D,\boldsymbol{P}}(\mathcal{S}_t)]\Big|\xi_{t-1} \Bigr\} \notag \\
    &\quad \quad + \sum_{t=1}^T\mathbb{P}(\bar{\xi}_{t-1})(n-K).
\end{align*}    
Select $\delta = \frac{1}{nT}$, then with probability at least $1-\frac{1}{nT}$, the good event $\xi_{t-1} = \Big \{|\boldsymbol{x}_e^T(\hat{\boldsymbol{\theta}}_{\tau-1} - \boldsymbol{\theta})| \leq \beta \sqrt{\boldsymbol{x}_e^T\boldsymbol{M}_{\tau-1}^{-1}\boldsymbol{x}_e},\forall e \in \mathcal{E}, \forall \tau \leq t \Big\}$ happens $\forall t \in \{1,2,\ldots, T\}$ with $\beta = \sigma^{-2}\sqrt{d\log(1+\frac{E^*T}{d}) + 2\log (nT)} + \sigma^{-2}\lambda E^cC + \Theta$ happens $\forall t \in [T]$. 

Besides, the \textsc{ORACLE} indicates that $f_{D,\hat{\boldsymbol{P}}_t}(\mathcal{S}_t) \geq \gamma \max_{\mathcal{S} \in \mathcal{V}}f_{D,\hat{\boldsymbol{P}}_t}(\mathcal{S})$ with probability at least $\alpha$. Thus, $\mathbb{E}[f_{D,\hat{\boldsymbol{P}}_t}(\mathcal{S}_t)] \geq \alpha \gamma \mathbb{E}[\max_{\mathcal{S} \in \mathcal{V}}f_{D,\hat{\boldsymbol{P}}_t}(\mathcal{S})]$. Under the good event $\xi_{t-1}$, we have $p(e) \leq \hat{p}_{\tau-1}(e)$, $\forall e \in \mathcal{E}$, $\forall \tau \leq t$. Thus, based on the monotonicity of $f$ in the probability weight, we have $\mathbb{E}[f_{D,\boldsymbol{P}}(\mathcal{S}^{opt})] \leq \mathbb{E} [f_{D,\hat{\boldsymbol{P}}_{t-1}}(\mathcal{S}^{opt})]$. Combining all the previous inequalities, we can obtain \cite{wen2017online,vaswani2017model}
\begin{align}
    &\mathbb{E}[f_{D,\boldsymbol{P}}(\mathcal{S}^{opt})] \leq \mathbb{E} [f_{D,\hat{\boldsymbol{P}}_{t-1}}(\mathcal{S}^{opt})] \notag \\
    & \quad \leq \mathbb{E} [\max_{\mathcal{S}: |\mathcal{S}| = K} f_{D,\hat{\boldsymbol{P}}_{t-1}}(\mathcal{S})] \leq \frac{1}{\alpha \gamma}\mathbb{E}[ f_{D,\hat{\boldsymbol{P}}_{t-1}}(\mathcal{S}_t)].
\end{align}
Therefore, under the good event $\xi_{t-1}$, $\forall t \in \{1, \ldots, T\}$ and according to the 1-Norm bounded smoothness condition, we have 
\begin{align}
    &\sum_{t=1}^T \Big [ f_{D,\boldsymbol{P}}(\mathcal{S}^{opt}) - \frac{1}{\alpha \gamma} f_{D,\boldsymbol{P}}(\mathcal{S}_t) \Big ] \notag \\
    &\leq \sum_{t=1}^T \Big [ \frac{1}{\alpha \gamma} f_{D,\hat{\boldsymbol{P}}_{t-1}}(\mathcal{S}_t) - \frac{1}{\alpha \gamma} f_{D,\boldsymbol{P}}(\mathcal{S}_t)  \Big ] \notag \\
    &\leq \frac{B}{\alpha \gamma}\sum_{t=1}^T \sum_{e \in \tilde{\mathcal{E}}_t}|\hat{p}_{e,t-1}-p_e| \notag \\
    &= \frac{B}{\alpha \gamma}  \sum_{t=1}^T \sum_{e \in \tilde{\mathcal{E}}_t}|\boldsymbol{x}_e^T(\hat{\boldsymbol{\theta}}_{t-1} - \boldsymbol{\theta}) + \beta\norm{\boldsymbol{x}_e}_{\boldsymbol{M}_{t-1}^{-1}}| \notag \\
    &\leq \frac{2B\beta}{\alpha \gamma} \sum_{t=1}^T \sum_{e \in \tilde{\mathcal{E}}_t} \sqrt{\boldsymbol{x}_e^T \boldsymbol{M}_{t-1}^{-1}\boldsymbol{x}_e}, \notag \\
    &= \frac{2B\beta}{\alpha \gamma}\Big [\underbrace{\sum_{t:\omega_{e,t}=1} \sum_{e \in \tilde{\mathcal{E}}_t} \sqrt{\boldsymbol{x}_e^T \boldsymbol{M}_{t-1}^{-1}\boldsymbol{x}_e}}_{I_1} \notag \\
    &\quad \quad + \underbrace{\sum_{t:\omega_{e,t}<1} \sum_{e \in \tilde{\mathcal{E}}_t} \sqrt{\boldsymbol{x}_e^T \boldsymbol{M}_{t-1}^{-1}\boldsymbol{x}_e}}_{I_2} \Big] \label{eq:rgt-i}
\end{align}
where $\tilde{\mathcal{E}}_t$ refers to the set of observed edges at time step $t$. 

\begin{definition}\label{def:obs}
For any time step $t$ and any directed edge $e \in \mathcal{E}$, we define the event
\begin{gather}
    O_t(e) = \{\text{edge $e$ is observed at round t}\}.
\end{gather}
%
\end{definition}
Based on the Definition~\ref{def:obs}, we have
\begin{gather}
    \sum_{e \in \tilde{\mathcal{E}}_t} \sqrt{\boldsymbol{x}_e^T \boldsymbol{M}_{t-1}^{-1}\boldsymbol{x}_e} = \sum_{e \in \mathcal{E}} \mathbbm{1}(O_t(e))\sqrt{\boldsymbol{x}_e^T \boldsymbol{M}_{t-1}^{-1}\boldsymbol{x}_e}.
\end{gather}
For the term $I_1$ defined in \eqref{eq:rgt-i}, we consider for all rounds $t \in [T]$, there exists $e \in \tilde{\mathcal{E}}_t$ with $\omega_{e,t} = 1$ and we assume these rounds can be listed as $\{t_1, t_2, \ldots, t_q\}$ for simplicity. With this notation, for each $i \leq q$, we can construct the auxiliary covariance matrix $\boldsymbol{A}_{i} = \boldsymbol{I} + \frac{1}{\sigma^2}\sum_{j=1}^{i}\sum_{e \in \mathcal{E}}\mathbbm{1}(O_{t_j}(e))\omega_{e,t_j}\boldsymbol{x}_{e}\boldsymbol{x}_{e}^T$ \cite{wen2017online}. According to the definition of matrix $\boldsymbol{M}_{t}$, we have 
\begin{gather}
    \boldsymbol{M}_{t_i} \succeq \boldsymbol{I} + \frac{1}{\sigma^2}\sum_{j=1}^{i}\sum_{e \in \mathcal{E}}\mathbbm{1}(O_{t_j}(e))\omega_{e,t_j}\boldsymbol{x}_{e}\boldsymbol{x}_{e}^T = \boldsymbol{A}_{u,i}.
\end{gather}
According to Lemma~\ref{lem:inv}, we have
\begin{gather}
    \boldsymbol{x}_{e}^T\boldsymbol{M}_{t_i}^{-1}\boldsymbol{x}_{e} \leq \boldsymbol{x}_{e}^T\boldsymbol{A}_{i}^{-1}x_{e}
\end{gather}
Therefore, the term $I_1$ defined in \eqref{eq:rgt-i} is bounded as shown in the following lemma. 
\begin{lemma}\label{lem:i_1}
For time step $t = 1, \ldots, T$, and $I_1$ as defined in \eqref{eq:rgt-i}, we have
\begin{gather}
    I_1 \leq \sqrt{\frac{TdE^* \log(1+\frac{TE^*}{d\sigma^2})}{\log(1+\frac{1^2}{\sigma^2})}}
\end{gather}
\end{lemma}
\begin{proof}
Define $z_{e,i} = \sqrt{\boldsymbol{x}_e^T \boldsymbol{A}_{i-1}^{-1}\boldsymbol{x}_e}$, and we have
\begin{align}
     \boldsymbol{A}_{i} = \boldsymbol{A}_{i-1} + \frac{1}{\sigma^2}\sum_{e \in \mathcal{E}}\mathbbm{1}(O_{t_{i}}(e))\boldsymbol{x}_e\boldsymbol{x}_e^T.   
\end{align}
For all $i \in \{1,2,\ldots,q\}$, $e \in \tilde{\mathcal{E}}_{t_{i}}$ ($e$ is observed at time step $t_{i}$), we have that
\begin{align}
    \det(\boldsymbol{A}_{i}) &\geq \det(\boldsymbol{A}_{i-1} + \frac{1}{\sigma^2}\boldsymbol{x}_e\boldsymbol{x}_e^T) \notag \\
    &= \det \left [\boldsymbol{A}_{i-1}^{\frac{1}{2}}\Big(\boldsymbol{I} + \frac{1}{\sigma}\boldsymbol{A}_{i-1}^{-\frac{1}{2}}\boldsymbol{x}_e\boldsymbol{x}_e^T\frac{1}{\sigma}\boldsymbol{A}_{i-1}^{-\frac{1}{2}}\Big)\boldsymbol{A}_{i-1}^{\frac{1}{2}}\right] \notag \\
    &= \det(\boldsymbol{A}_{i-1})\Big(1 + \frac{1}{\sigma^2}\boldsymbol{x}_e\boldsymbol{A}_{i-1}^{-1}\boldsymbol{x}_e^T\Big) \notag \\
    &= \det(\boldsymbol{A}_{u,i-1})\Big(1 + \frac{z_{e,i}^2}{\sigma^2}\Big).
\end{align}
Hence, we have
\begin{gather*}
    \det(\boldsymbol{A}_{i})^{|\tilde{\mathcal{E}}_{t_i}|} \geq \det(\boldsymbol{A}_{i-1})^{|\tilde{\mathcal{E}}_{t_i}|} \prod_{e \in \tilde{\mathcal{E}}_{t_i}}\Big(1 + \frac{z_{e,i}^2}{\sigma^2}\Big).
\end{gather*}

Denote $E^*$ as defined in Definition~\ref{def:Estar}, it is easy to obtain $|\tilde{\mathcal{E}}_{t_i}| \leq E^*$, $\forall i \in \{1,2,\ldots,q\}$ and we have
\begin{gather*}
    \det(\boldsymbol{A}_{t_q})^{E^*} \geq \det(\boldsymbol{A}_{0})^{E^*} \prod_{i=1}^{q} \prod_{e \in \tilde{\mathcal{E}}_{t_i}}\Big(1 + \frac{z_{e,i}^2}{\sigma^2}\Big).
\end{gather*}
$\boldsymbol{A}_{u,0} = \boldsymbol{I}$, $\forall u \in \mathcal{V}$. Besides, based on the algorithm, we have 
\begin{align}
    &\Tr(\boldsymbol{A}_{t_q}) = \Tr(\boldsymbol{I}+\frac{1}{\sigma^2}\sum_{i=1}^q\sum_{e \in \tilde{\mathcal{E}}_{t_i}}\boldsymbol{x}_e\boldsymbol{x}_e^T) \notag \\
    &= d + \frac{1}{\sigma^2}\sum_{i=1}^q\sum_{e \in \tilde{\mathcal{E}}_{t_i}}\norm{\boldsymbol{x}_e}^2 \leq d + \frac{TE^*}{\sigma^2},
\end{align}
the last inequality is based on the bound of $\norm{x_e} \leq 1$, $\forall e \in \mathcal{E}$ and $t_q \leq T$. According to the trace-determinant inequality, we have $\frac{1}{d}\Tr(A_{t_q}) \geq [\det(A_{t_q})]^{\frac{1}{d}}$, thus we have \cite{wen2017online}
\begin{align}
    \Big(1 + \frac{TE^*L^2}{d\sigma^2} \Big)^{dE^*} & \geq \Big[\frac{1}{d} \Tr(\boldsymbol{A}_{t_q})\Big]^{dE^*} \geq [\det(\boldsymbol{A}_{t_q})]^{E^*} \notag \\
    &\geq \prod_{i=1}^q \prod_{e \in \tilde{\mathcal{E}}_{t_{i-1}}}\Big(1 + \frac{z_{e,i}^2}{\sigma^2}\Big).
\end{align}
Take the logarithm on the both sides, we have
\begin{gather}
    dE^* \log(1+\frac{TE^*L^2}{d\sigma^2}) \geq \sum_{i=1}^q \sum_{e \in \tilde{\mathcal{E}}_{t_i}}\log\Big(1 + \frac{z_{e,i}^2}{\sigma^2}\Big),
\end{gather}
where $z_{e,i}^2 = \boldsymbol{x}_e^T\boldsymbol{A}_{i-1}^{-1}\boldsymbol{x}_e \leq \boldsymbol{x}_e^T\boldsymbol{A}_{u,0}^{-1}\boldsymbol{x}_e = \norm{x_e}^2 \leq 1$. And it is easy to prove that for any $y \in [0,1]$, we have $y \leq \frac{\log(1+\frac{y}{\sigma^2})}{\log(1+\frac{1}{\sigma^2})}$. Thus, we have $z_{e,i}^2 \leq \frac{\log(1+\frac{z_{e,i}^2}{\sigma^2})}{\log(1+\frac{1}{\sigma^2})}$. 

And then, it is satisfied that
\begin{align}
    \sum_{i=1}^q\sum_{e \in \tilde{\mathcal{E}}_{t_i}}z_{e,i}^2 \leq \frac{1}{\log(1+\frac{1}{\sigma^2})}&\sum_{i=1}^q\sum_{e \in \tilde{\mathcal{E}}_{t_i}}\log(1+\frac{z_{e,i}^2}{\sigma^2}) \notag \\
    &\leq \frac{dE^* \log(1+\frac{TE^*L^2}{d\sigma^2})}{\log(1+\frac{1}{\sigma^2})}.
\end{align}
Finally, with Cauthy-Schwarz inequality, we can obtain
\begin{align}
I_1 &= \sum_{t:\omega_{e,t} = 1}\sum_{e \in \mathcal{E}} \mathbbm{1}(O_t(e))\sqrt{\boldsymbol{x}_e^T \boldsymbol{M}_{t-1}^{-1}\boldsymbol{x}_e} \notag \\
&\leq \sum_{i=1}^q \sum_{e \in \mathcal{E}} \mathbbm{1}(O_{t_i}(e))\sqrt{\boldsymbol{x}_e^T \boldsymbol{M}_{t_{i-1}}^{-1}\boldsymbol{x}_e} \notag \\
&\leq \sum_{i=1}^q \sum_{e \in \mathcal{E}} \mathbbm{1}(O_{t_i}(e))\sqrt{\boldsymbol{x}_e^T \boldsymbol{A}_{i-1}^{-1}\boldsymbol{x}_e} \notag \\
&=  \sum_{i=1}^q\sum_{e \in \tilde{\mathcal{E}}_{t_{i-1}}}z_{e,i} \leq \sqrt{q E^*\Big[\sum_{i=1}^q\sum_{e \in \tilde{\mathcal{E}}_{t_i}}z_{e,i}^2\Big]}  \notag \\
&\leq E^*\sqrt{\frac{Td \log(1+\frac{TE^* L^2}{d\sigma^2})}{\log(1+\frac{1^2}{\sigma^2})}}.
\end{align}
That completes the proof of Lemma~\ref{lem:i_1}.
\end{proof}

For the second term $I_2$ defined in \eqref{eq:rgt-i}, according to the definition for weight $\omega_{e,t} < 1$, we have $\omega_{e,t} = \frac{\lambda}{\sqrt{\boldsymbol{x}_e^T\boldsymbol{M}_{t-1}^{-1}\boldsymbol{x}_e}}$, which implies that
\begin{align*}
    I_2 &= \sum_{t:\omega_{e,t}<1}\sum_{e \in \tilde{\mathcal{E}}_t} \sqrt{\boldsymbol{x}_e^T \boldsymbol{M}_{t-1}^{-1}\boldsymbol{x}_e} \notag \\
    &= \sum_{t:\omega_{e,t}<1}\sum_{e \in \tilde{\mathcal{E}}_t} \omega_{e,t}\boldsymbol{x}_e^T \boldsymbol{M}_{t-1}^{-1}\boldsymbol{x}_e/\lambda,
\end{align*}
where the second equation holds due to the definition of $\omega_{e,t}$. Now, we assume the rounds with weight $\omega_{e,t} < 1$ can be listed as $\{\tau_1,\ldots,\tau_k\}$. And we introduce the vector $\tilde{\boldsymbol{x}}_{e,i} = \sqrt{\omega_{e,\tau_i}}\boldsymbol{x}_e$ if at time step $\tau_i$, edge $(e$ is observable and matrix $\tilde{\boldsymbol{M}}_{i}$ as
\begin{gather}
    \tilde{\boldsymbol{M}}_{i} =  \boldsymbol{I} + \frac{1}{\sigma^2}\sum_{j=1}^{i}\omega_{e,\tau_j}\boldsymbol{x}_e\boldsymbol{x}_e^T = I + \sum_{j=1}^{i}\tilde{\boldsymbol{x}}_{e,i}\tilde{\boldsymbol{x}}_{e,i}^T.
\end{gather}
According to Lemma~\ref{lem:inv}, we also have $(\tilde{\boldsymbol{M}}_{i})^{-1} \succeq \boldsymbol{M}_{\tau_i}^{-1}$. Therefore, for each $i \in \{1,2,\ldots,k\}$, we have
\begin{gather}
    \boldsymbol{x}_e^T(\tilde{\boldsymbol{M}}_{i})^{-1} \boldsymbol{x}_e^T \geq  \boldsymbol{x}_e^T\boldsymbol{M}_{\tau_i}^{-1} \boldsymbol{x}_e^T.
\end{gather}
Follow the same process in the proof of Lemma~\ref{lem:i_1}, define $\tilde{z}_{e,i} = \sqrt{\tilde{\boldsymbol{x}}_{e,i}^T(\tilde{\boldsymbol{M}}_{i-1})^{-1}\tilde{\boldsymbol{x}}_{e,i}}$ and we have
\begin{gather}
    \tilde{\boldsymbol{M}}_{i} =  \tilde{\boldsymbol{M}}_{i-1} + \frac{1}{\sigma^2}\sum_{e\in \mathcal{E}}\mathbbm{1}(O_{\tau_i}(e))\tilde{\boldsymbol{x}}_{e,i}\tilde{\boldsymbol{x}}_{e,i}^T
\end{gather}
Similarly, for all $e \in \tilde{\mathcal{E}}_{\tau_i}$ and $i \in \{1,2,\ldots,k\}$, we can easily obtain
\begin{gather}
    \det(\tilde{\boldsymbol{M}}_{i}) \geq \det(\tilde{\boldsymbol{M}}_{i-1})(1+\frac{(\tilde{z}_{e,i})^2}{\sigma^2}),
\end{gather}
and 
\begin{gather*}
    \det(\tilde{\boldsymbol{M}}_{i})^{|\tilde{\mathcal{E}}_{\tau_i}|} \geq \det(\tilde{\boldsymbol{M}}_{i-1})^{|\tilde{\mathcal{E}}_{\tau_i}|} \prod_{e \in \tilde{\mathcal{E}}_{\tau_{i-1}}}\Big(1 + \frac{(\tilde{z}_{e,i-1})^2}{\sigma^2}\Big).
\end{gather*}
Follow the same process in Lemma~\ref{lem:i_1},
it is satisfied that
\begin{align}
    \sum_{i=1}^k\sum_{e \in \tilde{\mathcal{E}}_{\tau_i}}(\tilde{z}_{e,i})^2 &\leq \frac{1}{\log(1+\frac{1}{\sigma^2})}\sum_{i=1}^k\sum_{e\in \tilde{\mathcal{E}}_{\tau_i}}\log(1+\frac{(\tilde{z}_{e,i})^2}{\sigma^2}) \notag \\
    &\leq \frac{dE^* \log(1+\frac{TE^*}{d\sigma^2})}{\log(1+\frac{1}{\sigma^2})}.
\end{align}
And then
\begin{align}
    I_2 &= \sum_{t:\omega_{e,t} < 1}\sum_{e\in \mathcal{E}} \mathbbm{1}(O_t(e))\sqrt{\boldsymbol{x}_e^T \boldsymbol{M}_{t-1}^{-1}\boldsymbol{x}_e} \notag \\
    &= \sum_{t:\omega_{e,t} < 1}\sum_{e\in \mathcal{E}} \mathbbm{1}(O_t(e))\omega_{e,t}\boldsymbol{x}_e^T \boldsymbol{M}_{t-1}^{-1}\boldsymbol{x}_e/\lambda \notag \\
    &\leq \sum_{i=1}^k \sum_{e \in \mathcal{E}} \mathbbm{1}(O_{\tau_i}(e))\omega_{e,\tau_i}\boldsymbol{x}_e^T M_{\tau_{i-1}}^{-1}\boldsymbol{x}_e/\lambda \notag \\
    &\leq  \sum_{i=1}^k \sum_{e \in \mathcal{E}} \mathbbm{1}(O_{\tau_i}(e))\omega_{e,\tau_i}\boldsymbol{x}_e^T (\tilde{\boldsymbol{M}}_{i-1})^{-1}\boldsymbol{x}_e/\lambda \notag \\
    &= \sum_{i=1}^k \sum_{e \in \mathcal{E}} \mathbbm{1}(O_{\tau_i}(e))\tilde{\boldsymbol{x}}_{e,i}^T \tilde{\boldsymbol{M}}_{i-1}^{-1}\tilde{\boldsymbol{x}}_{e,i}/\lambda \notag  \\
    &= \sum_{i=1}^k\sum_{e \in \tilde{\mathcal{E}}_{\tau_i}} \frac{\tilde{z}_{e,i}^2}{\lambda} \leq \frac{dE^* \log(1+\frac{TE^*}{d\sigma^2})}{\lambda \log(1+\frac{1}{\sigma^2})}
\end{align}
Therefore, combine with the bound of $I_1$ in Lemma~\ref{lem:i_1}, we can obtain
\begin{align*}
&\sum_{t=1}^T\mathbb{E}\Bigl\{[f_{D,\boldsymbol{P}}(\mathcal{S}^{opt}) - \frac{1}{\alpha \gamma} f_{D,\boldsymbol{P}}(\mathcal{S}_t)]\Big|\xi_{t-1} \Bigr\}\notag \\
&\leq \frac{2B\beta}{\alpha \gamma}\Big (E^*\sqrt{\frac{Td \log(1+\frac{TE^*}{d\sigma^2})}{\log(1+\frac{1^2}{\sigma^2})}} + \frac{dE^* \log(1+\frac{TE^*}{d\sigma^2})}{\lambda \log(1+\frac{1}{\sigma^2})}\Big ).
\end{align*}

Select $\lambda = \frac{\sqrt{d}}{E^cC}$, the $\alpha \gamma$-scaled regret is upper bounded 
\begin{align}
     R^{\alpha \gamma}(T) &\leq \sum_{t=1}^T\mathbb{E}\Bigl\{[f_{D,\boldsymbol{P}}(\mathcal{S}^{opt}) - \frac{1}{\alpha \gamma} f_{D,\boldsymbol{P}}(\mathcal{S}_t)]\Big|\xi_{t-1} \Bigr\} \notag \\
     &\quad \quad + \sum_{t=1}^T\mathbb{P}(\bar{\xi}_{t-1})(n-K) \notag \\
     &\leq O(dBE^*\sqrt{T}\log(nT) + BE^*E^cCd\log(nT))
\end{align}
\end{proof}

\bibliographystyle{IEEEtran}
\bibliography{references} 
\end{document}